\documentclass[journal]{IEEEtran}
\usepackage[T1]{fontenc}
\usepackage{cite}
\usepackage{url}
\usepackage{amsmath,amssymb,amsfonts}
\usepackage{graphicx}
\usepackage{textcomp}
\usepackage[table]{xcolor}
\usepackage{bm}
\usepackage{amsthm}
\usepackage{enumerate}
\usepackage{booktabs}
\usepackage{diagbox}
\usepackage{stfloats}
\usepackage{supertabular}
\usepackage{subfigure}
\usepackage{caption}
\usepackage{algorithmicx}
\usepackage[linesnumbered,ruled,vlined]{algorithm2e}
\usepackage{comment}  
\usepackage{filecontents} 
\usepackage{flushend}
\usepackage{multirow}
\usepackage{array}
\definecolor{seagreen}{rgb}{0.18, 0.55, 0.34}
\definecolor{royalpurple}{rgb}{0.47,0.32,0.66}
\definecolor{brown(traditional)}{rgb}{0.59, 0.29, 0.0}
\definecolor{blue}{rgb}{0.3, 0.2, 0.9}
\usepackage[colorlinks,
            linkcolor=blue,
            anchorcolor=blue,
            citecolor=blue]{hyperref}

\newtheorem{theorem}{Theorem}

\newtheorem{remark}{Remark}

\DeclareMathOperator*{\argmin}{argmin}
\DeclareMathOperator*{\argmax}{argmax}
\begin{document} 

\title{
HybridRAG-based LLM Agents for Low-Carbon Optimization in Low-Altitude Economy Networks
}

\author{
Jinbo Wen, Cheng Su, Jiawen Kang, Jiangtian Nie, Yang Zhang, Jianhang Tang, \textit{Member, IEEE}, \\ Dusit Niyato, \textit{Fellow, IEEE}, and Chau Yuen, \textit{Fellow, IEEE}
\thanks{J. Wen and Y. Zhang are with the College of Computer Science and Technology, Nanjing University of Aeronautics and Astronautics, China (e-mails: jinbo1608@nuaa.edu.cn; yangzhang@nuaa.edu.cn). 
C. Su and J. Kang are with the School of Automation, Guangdong University of Technology, China (e-mails: chengsu9251@163.com; kavinkang@gdut.edu.cn).
J. Nie is with the School of Computer Science and Engineering, Nanyang Technological University, Singapore (e-mail: jnie001@e.ntu.edu.sg).
J. Tang is with the State Key Laboratory of Public Big Data, Guizhou University, China (e-mail: jhtang@gzu.edu.cn).
D. Niyato is with the College of Computing and Data Science, Nanyang Technological University, Singapore (e-mail: dniyato@ntu.edu.sg).
C. Yuen is with the School of Electrical and Electronics Engineering, Nanyang Technological University, Singapore (e-mail: chau.yuen@ntu.edu.sg).
}

}

\maketitle

\begin{abstract}
Low-Altitude Economy Networks (LAENets) are emerging as a promising paradigm to support various low-altitude services through integrated air-ground infrastructure. To satisfy low-latency and high-computation demands, the integration of Unmanned Aerial Vehicles (UAVs) with Mobile Edge Computing (MEC) systems plays a vital role, which offloads computing tasks from terminal devices to nearby UAVs, enabling flexible and resilient service provisions for ground users. To promote the development of LAENets, it is significant to achieve low-carbon multi-UAV-assisted MEC networks. However, several challenges hinder this implementation, including the complexity of multi-dimensional UAV modeling and the difficulty of multi-objective coupled optimization. To this end, this paper proposes a novel Retrieval Augmented Generation (RAG)-based Large Language Model (LLM) agent framework for model formulation. Specifically, we develop HybridRAG by combining KeywordRAG, VectorRAG, and GraphRAG, empowering LLM agents to efficiently retrieve structural information from expert databases and generate more accurate optimization problems compared with traditional RAG-based LLM agents. After customizing carbon emission optimization problems for multi-UAV-assisted MEC networks, we propose a Double Regularization Diffusion-enhanced Soft Actor-Critic (R\textsuperscript{2}DSAC) algorithm to solve the formulated multi-objective optimization problem. The R\textsuperscript{2}DSAC algorithm incorporates diffusion entropy regularization and action entropy regularization to improve the performance of the diffusion policy. Furthermore, we dynamically mask unimportant neurons in the actor network to reduce the carbon emissions associated with model training. Simulation results demonstrate the effectiveness and reliability of the proposed HybridRAG-based LLM agent framework and the R\textsuperscript{2}DSAC algorithm.
\end{abstract}


\section{Introduction}

With the advancement of wireless communication and aerial vehicle technologies, Low-Altitude Economy Networks (LAENets) are emerging as a transformative paradigm, enabling dynamic, large-scale, and intelligent communication infrastructure within low-altitude airspace below 1,000 meters~\cite{cai2025secure, cai2025large}. The primary goal of LAENets is to harness low-flying equipment with high mobility and operational flexibility capabilities, such as Unmanned Aerial Vehicles (UAVs), to carry out a wide range of economic activities, including logistics delivery, intelligent transportation, and environmental surveillance, thereby generating significant commercial and societal value~\cite{cai2025large}. As a critical architecture in LAENets, UAV-assisted Mobile Edge Computing (MEC) networks enable low-latency intelligent services for ground users through offloading computing tasks to proximate UAVs, mitigating the inherent limitations of resource-constrained mobile devices~\cite{SunTMC, 10381761}.

Driven by the prosperous vision of LAENets~\cite{cai2025secure, cai2025large}, achieving low-carbon multi-UAV-assisted MEC networks has garnered substantial interest from academia and industry~\cite{10677514}. The reduction of carbon emissions in multi-UAV-assisted MEC networks becomes a central focus in this pursuit~\cite{10439631}. However, several challenges hinder its implementation: 
\begin{itemize}
    \item \textit{Challenge I: Complexity of Multi-dimensional UAV Modeling.} In multi-UAV-assisted MEC networks, the complexity of UAV mathematical modeling primarily arises from four interrelated dimensions~\cite{wang2024multiuavenabledmecnetworks, 10381761, 10134570}: mobility (e.g., trajectory and velocity), communication (e.g., channel power gain and interference), and energy consumption (e.g., propulsion, computation, and transmission). Each dimension requires precisely calibrated modeling to ensure the accuracy and reliability of UAV system representations, posing a significant challenge for both interdisciplinary researchers and newcomers to this field.
    
    \item \textit{Challenge II: Difficulty of Multi-objective Coupled Optimization.} Unlike single-UAV-assisted MEC networks, the coupling of multiple objective variables, such as task offloading, resource allocation, and trajectory planning, significantly increases the difficulty of carbon emission optimization in multi-UAV-assisted MEC networks~\cite{10439631}. Some researchers have adopted heuristic algorithms to solve multi-objective coupled optimization problems for energy consumption reduction~\cite{SunTMC, 10677514, 10134570, 10439631}, but many of these problems are NP-hard, making it difficult for heuristic algorithms to find reasonable solutions~\cite{10679152}. As an alternative, Deep Reinforcement Learning (DRL) algorithms have been utilized to tackle these joint optimization issues~\cite{wang2024multiuavenabledmecnetworks, 10381761, 10197291}, but they are prone to falling into suboptimal solution exploration due to the high dimensionality of the action space.
\end{itemize}

Large Language Model (LLM) agents with Retrieval Augmented Generation (RAG) techniques have been developed to support network optimization tasks~\cite{10679152, 10815045}, which are capable of generating effective optimization strategies through interactive sessions with human users. Specifically, by comprehending the context and deep meaning of natural language, LLM agents can generate accurate content leveraging RAG techniques to retrieve relevant information from external databases~\cite{lewis2020retrieval,gao2023retrieval}. Although RAG can refine the outputs of LLMs and mitigate hallucination issues, it exhibits limitations when applied to carbon emission optimization in multi-UAV-assisted MEC networks. \textit{First}, traditional RAG techniques are often inefficient in capturing structural and relational information~\cite{peng2024graphretrievalaugmentedgenerationsurvey}, such as the influence of Line-of-Sight (LoS) probability on Ground-to-Air (G2A) communications, which cannot be represented through semantic similarity alone. \textit{Second}, traditional RAG techniques lack the capability to comprehensively grasp global contextual information, as they typically retrieve only a limited subset of documents, especially constrained to those appearing at the beginning or end of the documents~\cite{xiong2024graphmeetsretrievalaugmented}. 

Fortunately, GraphRAG, as an innovative extension of RAG, provides a promising solution to overcome the inherent limitations of RAG~\cite{peng2024graphretrievalaugmentedgenerationsurvey,xiong2024graphmeetsretrievalaugmented}. In contrast to traditional RAG, GraphRAG can capture intertextual relationships and retrieve graph elements enriched with relational knowledge from a constructed graph database~\cite{peng2024graphretrievalaugmentedgenerationsurvey}, enabling more accurate and efficient retrieval of structured relational information. Therefore, to address \textit{Challenge I}, we propose a HybridRAG-based LLM agent framework that synergistically integrates RAG and GraphRAG techniques. By leveraging the contextual knowledge retrieval capability of RAG and the structured reasoning capability of GraphRAG over interconnected parameters~\cite{gao2023retrieval, peng2024graphretrievalaugmentedgenerationsurvey}, the proposed HybridRAG-based LLM agents can efficiently retrieve expert knowledge and generate accurate carbon emission optimization problems for multi-UAV-assisted MEC networks. To address \textit{Challenge II}, by leveraging the ability of diffusion models to capture high-dimensional and intricate features within network environments~\cite{HongyangTMC, 10517486}, we propose a double regularization diffusion-enhanced DRL algorithm to generate optimal strategies of multi-objective carbon emission optimization problems formulated by HybridRAG-based LLM agents. Our contributions are summarized as follows:
\begin{itemize}
    \item \textbf{HybridRAG-based LLM Agent Framework:} We develop a HybridRAG-based LLM agent framework for carbon emission optimization in multi-UAV-assisted MEC networks. Specifically, we develop HybridRAG by merging KeywordRAG, VectorRAG, and GraphRAG. Through the blend of vector and keyword retrieval, LLM agents can not only assess vast embedded knowledge but also retrieve expert knowledge from external documents. Moreover, we construct a structured and queryable knowledge graph, enabling the LLM agents to retrieve structural relational information, thereby generating more accurate carbon emission optimization problems for multi-UAV-assisted MEC networks. (For \textit{Challenge I})
    \item \textbf{Double Regularization Diffusion-enhanced DRL:} We propose a Double Regularization Diffusion-enhanced Soft Actor-Critic (R\textsuperscript{2}DSAC) algorithm to identify optimal strategies of the carbon emission optimization problem formulated by the HybridRAG-based LLM agent. Specifically, we employ diffusion models as the policy to generate strategies through forward and reverse processes. To enhance policy performance, we incorporate diffusion entropy regularization and action entropy regularization into the policy learning objective function. Moreover, we apply dynamic pruning techniques in the diffusion-based actor network to suppress the activity of unimportant neurons, thereby reducing the carbon emissions associated with model training. (For \textit{Challenge II})
    \item \textbf{Extensive Performance Evaluation:} To evaluate the performance of the developed HybridRAG, we utilize a fine-grained framework called RAGChecker\footnote{\url{https://github.com/amazon-science/RAGChecker/tree/main}}. Specifically, we comprehensively adopt three types of metrics, including overall, retriever, and generator metrics, to rigorously analyze the performance of HybridRAG. Simulation results demonstrate that our HybridRAG outperforms traditional RAG techniques. In addition, we compare the proposed R\textsuperscript{2}DSAC algorithm with several DRL benchmark algorithms and conduct the ablation experiment. Simulation results demonstrate the effectiveness of the proposed R\textsuperscript{2}DSAC algorithm for carbon emission optimization in multi-UAV-assisted MEC networks.
\end{itemize}

The rest of the paper is organized as follows: Section \ref{related_work} reviews the related work. Section \ref{HybridRAG_framework} introduces the proposed HybridRAG-based LLM agent framework for carbon emission optimization in multi-UAV-assisted MEC networks. In Section \ref{RDMSAC}, we present the architecture of the R\textsuperscript{2}DSAC algorithm. Section \ref{Simulation_result} conducts extensive simulations to demonstrate the effectiveness of the proposed HybridRAG and the R\textsuperscript{2}DSAC algorithm. In Section \ref{Conclusion}, we conclude the paper. 

\section{Related Work}\label{related_work}
\subsection{Energy-Efficient Multi-UAV-Assisted MEC Networks}
In multi-UAV-assisted MEC networks, task offloading has emerged as a prominent research focus in recent studies~\cite{10381761, wang2024multiuavenabledmecnetworks, 10522499, 10032196}. An increasing number of researchers are dedicating their efforts to optimizing resource allocation during the task offloading process, aiming to achieve energy-efficient multi-UAV-assisted MEC networks~\cite{10381761, SunTMC, 10677514, wang2024multiuavenabledmecnetworks, 10134570}. For instance, the authors in~\cite{10381761} formulated a mixed integer programming problem by jointly optimizing the trajectory design of UAVs, offloading decisions, and computing and communication resource management. In~\cite{SunTMC}, the authors formulated a multi-objective optimization problem with respect to computing resource allocation, task offloading, and UAV trajectory control, thereby reducing the total task completion delay and UAV energy consumption. In~\cite{10677514}, the authors jointly optimized the flight trajectories of UAVs, computing resource allocation, and task offloading scheduling, thereby minimizing the carbon emissions of blockchain-enabled UAV-assisted MEC systems.

However, the aforementioned studies typically construct intricate optimization problems manually, introducing a susceptibility to human errors that can compromise the precision of the devised strategy~\cite{10679152}. Moreover, a significant proportion of current researchers neglect the investigation of reducing the carbon emissions associated with multi-UAV-assisted MEC networks, which is a critical aspect in the context of sustainable low-altitude economies. To address this research gap, we propose leveraging LLM agents as auxiliary tools to formulate carbon emission optimization problems for multi-UAV-assisted MEC networks, ultimately enabling sustainable and low-carbon operation in such systems.

\subsection{RAG-based LLM Agents in Network Optimization}
As a core component of agentic Artificial Intelligence (AI), LLM agents are capable of creating novel text through \textit{chain-of-thought} reasoning~\cite{10.1145/3677052.3698671, lee2024hybgraghybridretrievalaugmentedgeneration, 10750803}. Empowered by RAG, which is an advanced technique for enhancing the reliability and accuracy of Generative AI (GenAI) models~\cite{gao2023retrieval}, LLM agents can generate highly accurate outputs by retrieving relevant contextual information from external databases. Some researchers have explored the applications of RAG-based LLM agents in network optimization~\cite{10679152, 10815045}. Specifically, the authors in~\cite{10679152} proposed a RAG-based GenAI agent framework to customize problem formulation in satellite communication networks. In~\cite{10815045}, the authors proposed an LLM-enabled carbon emission optimization framework for task offloading, which designed pluggable LLM and RAG modules to generate accurate and reliable carbon emission optimization problems. While RAG has demonstrated promise in static information retrieval tasks~\cite{10679152, 10815045, gao2023retrieval}, its performance often degrades in complex and dynamic optimization scenarios that require connecting disparate pieces of information. Therefore, we develop HybridRAG by integrating GraphRAG, KeywordRAG, and VectorRAG, thereby enhancing the robustness and adaptability of LLM agents to complex network optimization scenarios.

\subsection{Diffusion-based DRL Algorithms}
Diffusion models have gained significant research attention in DRL due to their exceptional expressiveness in policy representation and inherent multimodality for capturing diverse solution spaces~\cite{wang2023diffusionpoliciesexpressivepolicy, NEURIPS2023_d45e0bfb, HongyangTMC, wen2025diffusion}. For instance, the authors in~\cite{wang2023diffusionpoliciesexpressivepolicy} leveraged a conditional diffusion model to represent the Q-learning policy. In~\cite{NEURIPS2023_d45e0bfb}, the authors proposed an efficient diffusion policy that approximately constructs actions from corrupted ones during training. Diffusion-based DRL algorithms have been utilized for network optimization~\cite{HongyangTMC, 10815045, 10829636, 10841385}. For instance, the authors in~\cite{10815045} employed diffusion-based DRL algorithms to identify optimal carbon emission strategies for task offloading. In~\cite{10829636}, the authors proposed a multi-dimensional contract design policy based on diffusion models to generate optimal contracts, thereby enhancing the efficiency of embodied agent AI twin migrations. Inspired by the successful applications of diffusion models in network optimization~\cite{HongyangTMC}, we adopt diffusion-based DRL algorithms to solve the optimization problem formulated by the proposed HybridRAG-based LLM agent framework.

\section{HybridRAG-based LLM Agent Framework}\label{HybridRAG_framework}

In this section, we propose the HybridRAG-based LLM agent framework for carbon emission optimization in multi-UAV-assisted MEC networks.

\subsection{Expert Dataset Introduction}

\begin{figure}
    \centering
    \includegraphics[width=0.48\textwidth]{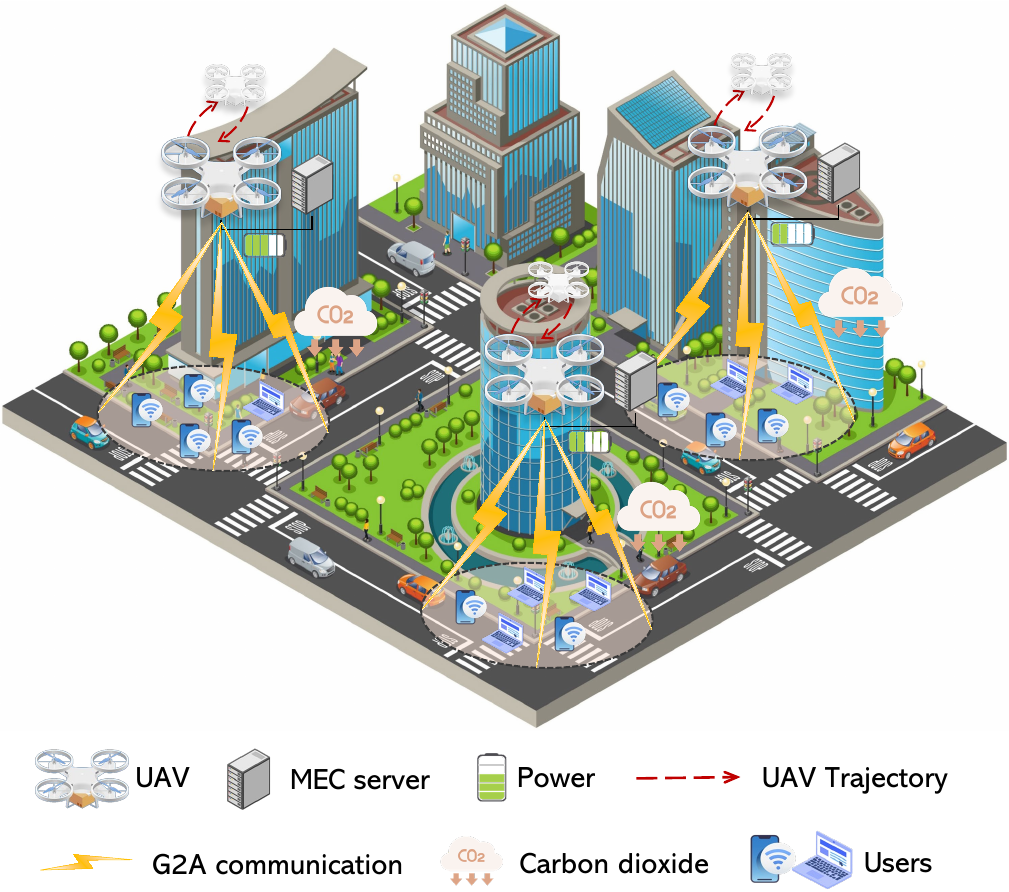}
    \caption{An illustration of LAENets, where UAVs emit carbon dioxide during computing service provisions.}
    \label{System_model}
\end{figure}

\subsubsection{System model}
As shown in Fig. \ref{System_model}, we consider a multi-UAV-assisted MEC network that consists of a set $\mathcal{M} = \{1,\ldots, m, \ldots, M\}$ of $M$ rotary-wing UAVs and a set $\mathcal{K} = \{1, \ldots, k, \ldots, K\}$ of $K$ users. Without loss of generality, the multi-UAV-assisted MEC network is under a continuous duration~\cite{9930881}, which is divided into a set $\mathcal{N} = \{1,\ldots,n,\ldots
,N\}$ of $N$ time slots with an equal time duration $\delta_t$. At each time slot, each user $k$ generates an indivisible computing-intensive task $\boldsymbol{I}_k(n) = (D_k(n), C_k(n))$~\cite{wang2024multiuavenabledmecnetworks, 10677514}, where $D_k(n)$ ($\rm{bits}$) represents the task size and $C_k(n)$ ($\rm{cycles/bit}$) represents the number of Central Processing Unit (CPU) cycles required for computing one bit of task data.

Due to the limited computing ability of user devices, each user can offload the generated task (e.g., healthy data analysis) to a rotary-wing UAV equipped with one or multiple MEC servers~\cite{10677514, SunTMC}. Compared with fixed-wing UAVs, rotary-wing UAVs can hover near user devices to provide computing services for users with better channel quality~\cite{9930881}. In addition, rotary-wing UAVs are generally equipped with multiple antennae~\cite{10381761}, and multiple users can communicate with UAVs through orthogonal frequency division multiple access at the same time~\cite{wang2024multiuavenabledmecnetworks}, indicating that interference between multiple users in the coverage area of a UAV can be ignored~\cite{wang2024multiuavenabledmecnetworks}.

\subsubsection{UAV mobility model}
We consider a three-dimensional Cartesian coordinate system in the multi-UAV-assisted MEC network, where the coordinate of each user $k$ at time slot $n$ is denoted as $\mathbf{v}_k(n) = [x_k(n), y_k(n), 0]^T$. In addition, we consider that UAVs hover at a certain height $H$~\cite{9930881}, and the coordinate of each UAV $m$ at time slot $n$ is denoted as $\mathbf{w}_m(n) = [x_m(n), y_m(n), H]^T$. To ensure that UAVs operate within the designated rectangular area, their positions must satisfy the following constraints, which are given by~\cite{ZhaoTWC}
\begin{equation}
    x_m(n) \in [0, X_{\rm{max}}],\: y_m(n)\in [0, Y_{\rm{max}}],\: \forall m \in \mathcal{M}, \forall n \in \mathcal{N},
\end{equation}
where $X_{\rm{max}}$ and $Y_{\rm{max}}$ are the side lengths of the designated rectangular area, respectively. Notably, the next positions of UAV $m$ is determined by its horizontal angle $\theta_m(n) \in [0, 2\pi)$ and its flight speed $v_m(n)$, and we can obtain
\begin{equation}
    x_m(n+1) = x_m(n) + \delta_t v_m(n) \cos(\theta_m(n)),\:\forall n \in \mathcal{N},
\end{equation}
\begin{equation}
    y_m(n+1) = y_m(n) + \delta_t v_m(n) \sin(\theta_m(n)),\:\forall n \in \mathcal{N}.
\end{equation}

The mobility model of UAVs is constrained by the maximum instantaneous speed $V_{\rm{max}}$ and the minimum collision avoidance distance $D_{\rm{min}}$~\cite{wang2024multiuavenabledmecnetworks, 10381761}, which are expressed as
\begin{equation}
    \frac{\lVert\mathbf{w}_m(n+1) - \mathbf{w}_m(n) \rVert}{\delta_t} \leq V_{\rm{max}},\: \forall m \in \mathcal{M}, \forall n \in \mathcal{N},
\end{equation}
\begin{equation}
    \lVert \mathbf{w}_m(n) - \mathbf{w}_{m^{\prime}}(n) \rVert \geq D_{\rm{min}},\: \forall m, m^{\prime} \in \mathcal{M}, m \neq m^{\prime},
\end{equation}
where $\lVert \cdot \rVert$ represents the Euclidean norm~\cite{9930881}. When user $k$ is within the coverage region of UAV $m$, the task $\boldsymbol{I}_k(n)$ generated by user $k$ can be offloaded to UAV $m$ at each time slot $n$, and the corresponding constraint is expressed as
\begin{equation}
    \lVert \mathbf{w}_m(n) - \mathbf{v}_k(n)\rVert^2 \leq r_{\rm{max}}^2 + H^2,\: \forall m \in \mathcal{M}, \forall k \in \mathcal{K},
\end{equation}
where $r_{\rm{max}}$ represents the coverage area radius of UAVs.

\subsubsection{Communication model}
In real environments, many different objects act as scatterers or obstacles during G2A communications~\cite{10134570}. Radio signals emitted by UAVs or user devices do not propagate in free space but may be affected by scattering or shadowing caused by objects~\cite{10134570}, resulting in additional path loss. Thus, we capture the G2A communication by a probabilistic path loss model rather than a simplified free path loss model~\cite{10381761, wang2024multiuavenabledmecnetworks, 10134570}. Specifically, the probabilistic path loss model considers the occurrence probability of each path and the path loss of LoS and Non-LoS (NLoS) links. The occurrence probabilities of LoS and NLoS communications between user $k$ and UAV $m$ at time slot $n$ are expressed as~\cite{10134570}
  \begin{align} 
    &P_{k,m}^{\rm{LoS}}(n) = \frac{1}{1 + a\exp(-b(\theta_{k,m}(n)-a))},\\
    &P_{k,m}^{\rm{NLoS}}(n) = 1 - P_{k,m}^{\rm{LoS}}(n),
  \end{align}
where $a$ and $b$ are constant parameters related to the environment, and $\theta_{k,m}(n) = \frac{180}{\pi}\arcsin\big(\frac{H}{\lVert \mathbf{w}_m(n) - \mathbf{v}_k(n)\rVert}\big)$ represents the elevation angle of user $k$ to UAV $m$. 

In the probabilistic path loss model, Free-Space Path Loss (FSPL) is the foundation of the signal attenuation in ideal conditions, and the FSPL for the communication between user $k$ and UAV $m$ is given by
\begin{equation}
    L_{k,m}(n) = 20\Big[\lg(\lVert \mathbf{w}_m(n) - \mathbf{v}_k(n)\rVert) + \lg(f) +  \lg\Big(\frac{4\pi}{c}\Big)\Big],
\end{equation}
where $f$ is the carrier frequency and $c = 3 \times 10^8\:\rm{m/s}$ is the speed of light. Thus, the average path loss for the G2A communication between user $k$ and UAV $m$ is given by~\cite{10381761}
\begin{equation}
\begin{split}
    \overline{PL}_{k,m}(n)(\rm{dB}) = &P_{k,m}^{\rm{LoS}}(n)(L_{k,m}(n) + \eta^{\rm{LoS}})\\
    &+ P_{k,m}^{\rm{NLoS}}(n)(L_{k,m}(n) + \eta^{\rm{NLoS}}),
\end{split}
\end{equation}
where $\eta^{\rm{LoS}}$ and $\eta^{\rm{NLoS}}$ are the excessive path loss for LoS and NLoS links, respectively. Therefore, the uplink transmission rate from user $k$ to UAV $m$ at time slot $n$ is given by~\cite{wang2024multiuavenabledmecnetworks}
\begin{equation}
    R_{k,m}(n) = B_{k,m}^{\rm{G2A}}\log_2\bigg(1+\frac{p_k(n)}{\overline{PL}_{k,m}(n)\delta^2}\bigg),
\end{equation}
where $B_{k,m}^{\rm{G2A}}$ is the uplink bandwidth from user $k$ to UAV $m$, $p_k(n)$ is the transmit power of user $k$, and $\delta^2$ is the additive Gaussian white noise power at the receiving end. Additionally, the total uplink bandwidth from users to UAV $m$ cannot exceed the maximum uplink bandwidth of G2A communications~\cite{wang2024multiuavenabledmecnetworks}, and the corresponding constraint is given by
\begin{equation}
    \sum_{k = 1}^K B_{k,m}^{\rm{G2A}} \leq B_{\rm{max}}^{\rm{G2A}},\: \forall m \in \mathcal{M}.
\end{equation}

The size of uploaded tasks is generally larger than the size of calculation results~\cite{10677514}. Therefore, without loss of generality, we do not consider the latency and energy consumption of downlink transmission~\cite{10677514, SunTMC}.

\subsubsection{Computing energy consumption model} The computing energy consumption model consists of two parts, i.e., the energy consumption for uplink transmission and the energy consumption for UAV calculating offloading tasks. The energy consumption for uplink transmission from user $k$ to UAV $m$ at time slot $n$ is expressed as~\cite{10381761}
\begin{equation}
    E_{k,m}^{\rm{Tran}} (n) = \frac{p_k(n)D_k(n)}{R_{k,m}(n)}.
\end{equation}
In addition, the energy consumption for UAV $m$ calculating the task $\boldsymbol{I}_k(n)$ offloaded by user $k$ at time slot $n$ is given by
\begin{equation}
    E_{k,m}^{\rm{Cal}} (n) = \epsilon_m D_k(n) C_k(n) (f_{k,m}(n))^2,
\end{equation}
where $\epsilon_m$ denotes the effective switching capacitance that depends on the CPU architecture of the MEC server installed on UAV $m$~\cite{SunTMC}, and $f_{k,m}(n)$ ($\rm{cycles/s}$) represents the computing capability allocated by UAV $m$ to user $k$ at time slot $n$. 

We denote $\alpha_{k,m}(n)$ as a task offloading variable and consider that each user only offloads one task to a UAV at each time slot~\cite{10677514}, and the corresponding constraint is given by
\begin{equation}
    \sum_{m=1}^M \alpha_{k,m}(n) = 1,\: \forall k \in \mathcal{K}, n \in \mathcal{N},
\end{equation}
where $\alpha_{k,m}(n) = 1$ indicates that user $k$ offloads a computing task to UAV $m$ at time slot $n$, and otherwise $\alpha_{k,m}(n) = 0$.

\subsubsection{UAV propulsion consumption model}
The propulsion energy consumption of UAV $m$ flying at speed $v_m(n)$ at time slot $n$ is expressed as~\cite{10677514, 9930881, SunTMC}
\begin{equation}
\begin{split}
   E_{m}^{\rm{Fly}} (n) = &\delta_t \Bigg[ P_0\bigg(1+\frac{3(v_m(n))^2}{(U^{\rm{Tip}})^2}\bigg) + \frac{1}{2} d_0 \zeta s A (v_m(n))^3\\
   &+ P_1\sqrt{\sqrt{1+\frac{(v_m(n))^4}{4v_0^4}}-\frac{(v_m(n))^2}{2v_0^2}}
   \Bigg],
\end{split}
\end{equation}
where $P_0$ and $P_1$ represent the blade profile power and the induced power for hovering, respectively, $U^{\rm{Tip}}$ represents the blade tip speed of the rotor blade, $d_0$ represents the UAV fuselage drag ratio, $\zeta$ represents the air density, $s$ represents the rotor solidity, $A$ represents the rotor disk area, and $v_0$ represents the average rotor induced speed.

\subsubsection{Optimization model}
The total energy consumption of the multi-UAV-assisted MEC network is expressed as
\begin{equation}
\begin{split}
  E^{\rm{Total}} = &\sum_{n=1}^N\sum_{m=1}^M\sum_{k=1}^K \alpha_{k,m}(n) \Big[E_{k,m}^{\rm{Tran}} (n) + E_{k,m}^{\rm{Cal}} (n)\Big]\\
  &+\sum_{n=1}^N\sum_{m=1}^M E_{m}^{\rm{Fly}} (n).
\end{split}
\end{equation}

The optimization model in this paper is to minimize the carbon emissions of the multi-UAV-assisted MEC network, thereby promoting LAENets. According to~\cite{10677514, GOODCHILD201858}, the carbon emissions of the multi-UAV-assisted MEC network are related to the total energy consumption of the multi-UAV-assisted MEC network, which can be expressed as
\begin{equation}
    C^{\rm{Total}} = \varsigma^{\rm{Carbon}} \cdot \tau \cdot E^{\rm{Total}},
\end{equation}
where $\varsigma^{\rm{Carbon}} = 3.773 \times 10^{-4}$ represents that the power generation equipment of a UAV will generate $3.773 \times 10^{-4}\: \rm{kg}$ when consuming $1$ Watt-hour (Wh) of electricity, and $\tau$ denotes the conversion coefficient between Whs and joules.

\begin{remark}
Network designers need to carefully configure the above six models to manually construct a reliable carbon emission optimization problem for multi-UAV-assisted MEC networks~\cite{10679152}. Nevertheless, for network designers unfamiliar with UAVs or carbon emission reduction, there may be some human errors in formulating the carbon emission optimization problem~\cite{10679152, 10815045}, such as ignoring UAV flight propulsion consumption or the conversion coefficient between Whs and joules, affecting the accuracy of the strategies derived from the constructed optimization problem for reducing carbon emissions in the multi-UAV-assisted MEC network.
\end{remark}

RAG-based LLM agents have been developed to assist network designers in formulating accurate network optimization problems~\cite{10815045, 10679152}, which can generate optimization problems through multiple interactions with network designers. Although traditional RAG excels in quickly generating coherent responses from related textual documents~\cite{10.1145/3677052.3698671}, it may lose critical contextual information due to paragraph-level chunking. Fortunately, GraphRAG, as an innovative RAG technique, takes into account the interconnections between textual documents, allowing for more accurate and comprehensive retrieval of relational information~\cite{xiong2024graphmeetsretrievalaugmented}. However, GraphRAG enhances contextual understanding while resulting in higher token and time costs compared with traditional RAG. Moreover, both traditional RAG and GraphRAG cannot effectively perform \textit{hybrid question} answering~\cite{lee2024hybgraghybridretrievalaugmentedgeneration}, especially in complex network optimization scenarios, where hybrid questions require both textual and relational information to be answered correctly~\cite{lee2024hybgraghybridretrievalaugmentedgeneration}. For example, in the multi-UAV-assisted MEC network, ``\textit{the data transmission rate between user $k$ and UAV $m$ at time slot $t$}'' is the textual information, and ``\textit{the G2A communication link}'' is the relational information. To this end, we develop a HybridRAG-based LLM agent by merging KeywordRAG, VectorRAG, and GraphRAG. 

\begin{remark}
By combining the advantages of both traditional RAG and GraphRAG, the HybridRAG-based LLM agent can effectively analyze the retrieved expert knowledge and accurately generate the configurations of each model based on the hybrid questions from network designers. In particular, the developed HybridRAG-based LLM agent framework can adapt to various network optimization tasks through the corresponding external databases.
\end{remark}

\begin{figure}
    \centering
    \includegraphics[width=0.48\textwidth]{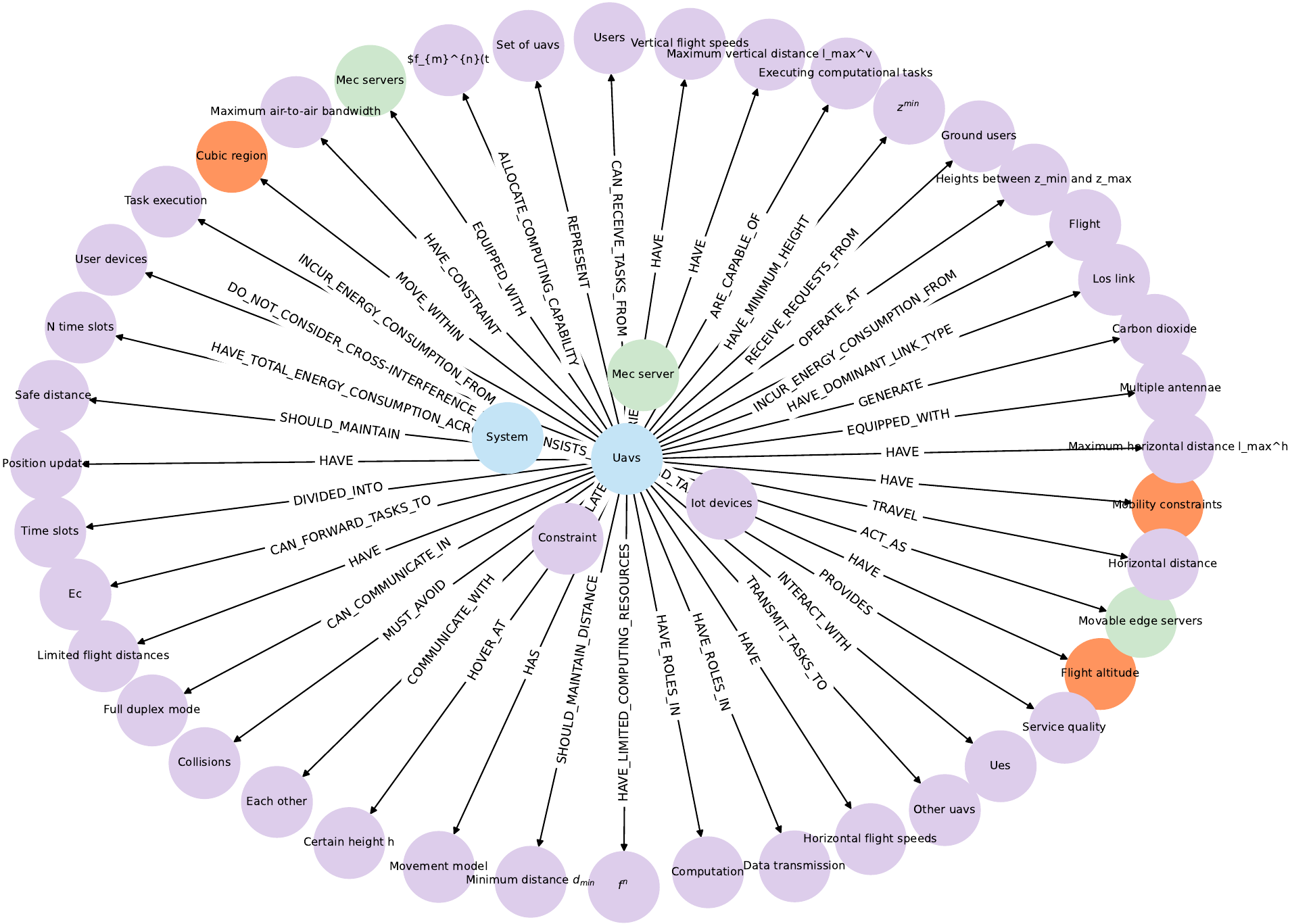}
    \caption{A structured and queryable knowledge graph for the formulation of carbon emission optimization problems. We construct the knowledge graph based on expert data consisting of academic papers from IEEE Xplore, involving carbon emission reduction, resource allocation, and task offloading in multi-UAV-assisted MEC networks.
    }
    \label{Knowledge_graph}
\end{figure}
\subsection{HybridRAG-based LLM Agents}

The HybridRAG comprises three core modules: KeywordRAG, VectorRAG, and GraphRAG~\cite{10.1145/3677052.3698671}. In the following, we introduce the working principle of HybridRAG.

\subsubsection{KeywordRAG}
KeywordRAG, as the key component of the HybridRAG framework, enables LLM agents to retrieve expert knowledge from external documents via keyword matching, instead of relying on semantic search~\cite{ren2024retrieval}. This functionality is especially crucial in wireless communication contexts characterized by dense and nuanced domain-specific terminology. To enhance the performance of KeywordRAG, we align expert-curated keyword segments with domain-specific terms extracted from optimization-related queries, facilitating faster, more precise, and more efficient retrieval of relevant knowledge passages from external databases.

We first perform semantic segmentation on the original external dataset $\mathcal{D}$ and extract hierarchical heading structures from the source corpus. The extracted heading structures function as contextual keywords and serve as hierarchical indices for the segmented text units, collectively forming the retrieval database $\mathcal{D}^{\prime}$~\cite{lee2024hybgraghybridretrievalaugmentedgeneration}. During the interaction between the network designer and the LLM agent regarding the formulation of carbon emission optimization problems, the LLM agent extracts the relevant keyword set $\boldsymbol{W}$ from input queries $\boldsymbol{Q}$ and performs keyword-based matching against the index structure in $\mathcal{D}'$, which is expressed as
\begin{equation} 
    \begin{split} 
        \boldsymbol{W} &= \mathcal{F}_{\rm{keyword}}(\boldsymbol{Q}),\\
        \boldsymbol{W}' &= \{w_i \in \boldsymbol{W} \mid w_i \in \boldsymbol{W}_{\rm{index}}\}, 
    \end{split} 
\end{equation}
where $\mathcal{F}_{\rm{keyword}}(\cdot)$ represents a prompt-based LLM function designed to extract keywords from input queries $\boldsymbol{Q}$, $\boldsymbol{W}_{\rm{index}}$ represents the pre-defined keyword set associated with the indexed corpus, and $\boldsymbol{W}'$ represents the filtered set of valid keywords used for matching.

After keyword filtering, the retrieved documents are ranked based on keyword frequency and relevance~\cite{10.1145/3677052.3698671}. The weight of each document $\boldsymbol{c}_j$ is defined as
\begin{equation} 
    \begin{split} 
        \kappa(\boldsymbol{c}_j) = \sum_{w_i \in \boldsymbol{W}'} \mathbb{I}(w_i \in \mathcal{T}(\boldsymbol{c}_j)), 
    \end{split} 
\end{equation}
where $\mathcal{C} = \{\boldsymbol{c}_1, \ldots, \boldsymbol{c}_j, \ldots, \boldsymbol{c}_J\}$ is the set of all documents in $\mathcal{D}'$, $\mathcal{T}(\boldsymbol{c}_j)$ represents the set of keywords associated with document $\boldsymbol{c}_j$, and $\mathbb{I}(\cdot)$ denotes the indicator function, which yields $1$ if the condition is satisfied and $0$ otherwise. Then, the top $G_{\rm{top}}$ ranked documents are selected and returned to the LLM agent, which is given by
\begin{equation} 
    \mathcal{C}_{\rm{keyword}} = \argmax\sum_{\boldsymbol{c}_{j}\in\mathcal{C}}\kappa(\boldsymbol{c}_{j}),\:|\mathcal{C}_{\rm{keyword}}| = G_{\rm{top}},
\end{equation}
where $\mathcal{C}_{\rm{keyword}}$ denotes the final set of selected documents and $G_{\rm{top}}$ represents the maximum number of text blocks retrieved for downstream processing. The retrieved passages are subsequently concatenated with the initial query and forwarded to the LLM agent, enabling the generation of responses that are contextually aligned with the carbon emission optimization setting in multi-UAV-assisted MEC networks~\cite{gao2023retrieval, lee2024hybgraghybridretrievalaugmentedgeneration}. 


\begin{figure*}[t]
    \centering
    \vspace{-0.5em}
    \includegraphics[width=0.98\textwidth]{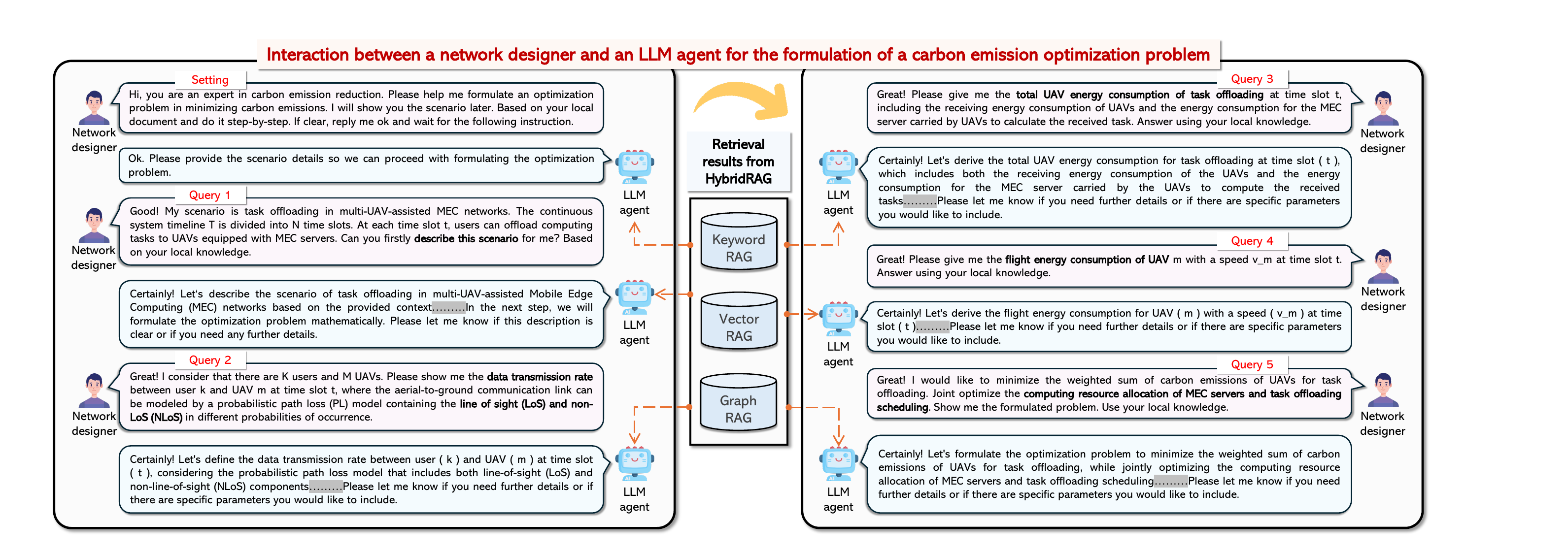}
    \caption{The formulation process of a carbon emission optimization problem for multi-UAV-assisted MEC networks by using the proposed HybridRAG-based LLM agent. The corresponding code and the whole formulation process can be referred to \url{https://github.com/secretcheng/HybridRAG-for-Network-Optimization}.}
    \label{Interaction_user_agent}
\end{figure*}

\subsubsection{GraphRAG}
GraphRAG is an advanced technique that incorporates knowledge graphs into traditional RAG, enabling structured semantic representation and relational reasoning~\cite{xiong2024graphmeetsretrievalaugmented, peng2024graphretrievalaugmentedgenerationsurvey}. Within the context of multi-UAV-assisted MEC networks, GraphRAG can capture the complex relationships among domain-specific entities by leveraging the structured semantics of knowledge graphs~\cite{xiong2024graphmeetsretrievalaugmented}. These entities include ``\textit{UAVs},'' ``\textit{MEC servers},'' ``\textit{User devices},'' ``\textit{Computing tasks},'' and ``\textit{Network resources},'' and their interrelations can be modeled through semantic triplets such as ``\textit{UAVs equipped with MEC servers},'' ``\textit{User devices offload computing tasks to UAVs},'' and ``\textit{MEC servers allocate computing resources}.''

We first utilize LLMs to construct a domain-specific knowledge graph from expert data, as shown in Fig. \ref{Knowledge_graph}. By systematically analyzing the interconnections among entities, GraphRAG enhances logical consistency and precision in the formulation of optimization objectives and constraints. The GraphRAG pipeline comprises three key stages:
\begin{itemize}
    \item \textbf{Knowledge Graph Construction:} Building upon the retrieval database $\mathcal{D}'$, we define a triplet extraction task facilitated by LLMs. Each unstructured document in $\mathcal{D}'$ is processed to extract relevant entities and their relationships~\cite{xiong2024graphmeetsretrievalaugmented, 10.1145/3677052.3698671}, which is expressed as
    \begin{equation}
    \begin{split}
        \boldsymbol{G} &= \mathcal{F}_{\rm{triplet}}(\boldsymbol{T})\\
        &= \{(s_1, p_1, o_1), \ldots, (s_I, p_I, o_I)\},  
    \end{split}
   \end{equation}
    where $\boldsymbol{G}$ represents the set of semantic triplets extracted from a text segment $\boldsymbol{T} \in \mathcal{D}'$ and $\mathcal{F}_{\rm{triplet}}(\cdot)$ represents a prompt-based LLM function designed for triplet extraction. Each triplet $(s_i, p_i, o_i)$ consists of a subject $s_i$, a predicate $p_i$, and an object $o_i$. All extracted triplets are organized based on the subject-object relationships and stored in the Neo4j graph database\footnote{\url{https://github.com/neo4j/neo4j}}, providing a structured and queryable knowledge graph for downstream reasoning and retrieval tasks.

    \item \textbf{Query Processing:} Upon receiving a new input query $\boldsymbol{Q}$, GraphRAG first parses $\boldsymbol{Q}$ to identify relevant entities and their synonyms, as well as potential relationships among them, which is expressed as
    \begin{equation} 
        \boldsymbol{E} = \mathcal{F}_{\rm{entity}}(\boldsymbol{Q}),\:\boldsymbol{E}' = \mathcal{F}_{\rm{synonym}}(\boldsymbol{E}), 
    \end{equation}
    \begin{equation}
        \boldsymbol{E}_{\rm{final}} = \boldsymbol{E} \cup \boldsymbol{E}',
    \end{equation}
    where $\mathcal{F}_{\rm{entity}}(\cdot)$ and $\mathcal{F}_{\rm{synonym}}(\cdot)$ represent prompt-based LLM functions responsible for entity recognition and synonym expansion, respectively. The final entity set $\boldsymbol{E}_{\rm{final}}$ is obtained by taking the union of the recognized entities $\boldsymbol{E}$ and their corresponding synonyms $\boldsymbol{E}'$. Then, the system traverses the constructed knowledge graph to identify paths connecting the entities in $\boldsymbol{E}_{\rm{final}}$, thereby uncovering the semantic and structural relationships in the domain~\cite{xiong2024graphmeetsretrievalaugmented, peng2024graphretrievalaugmentedgenerationsurvey}.
    \item \textbf{Information Retrieval and Generation:} Based on the paths discovered within the knowledge graph, GraphRAG retrieves semantically relevant and structurally coherent information, which is expressed as
    \begin{equation} 
        \mathcal{C}_{\rm{graph}} = \mathrm{str}(\mathcal{H}_{\rm{graph}}(\boldsymbol{E}_{\rm{final}}, d)), 
    \end{equation}
    where $\mathcal{C}_{\rm{graph}}$ represents the set of semantic triplets retrieved from the knowledge graph, $\mathcal{H}_{\rm{graph}}(\cdot)$ is the subgraph retrieval function, $\mathrm{str}(\cdot)$ is the function that translates the retrieved subgraph data into string, and $d$ represents the depth of graph traversal. Due to the structured nature of the knowledge graph, the retrieved information inherently captures the mathematical and logical dependencies among domain-specific variables~\cite{xiong2024graphmeetsretrievalaugmented}. The retrieved subgraph data is then provided as input to the LLM agent, alongside the original query $\boldsymbol{Q}$, to generate a context-sensitive and domain-aligned response. 
\end{itemize}

\subsubsection{HybridRAG}
The proposed HybridRAG technique integrates the advantages of KeywordRAG, VectorRAG, and GraphRAG~\cite{10.1145/3677052.3698671}, enabling effective adaptation to various query types and data sources~\cite{su2024hybrid}, which is particularly suitable for the dynamic and heterogeneous characteristics of multi-UAV-assisted MEC networks.

Building upon the outputs of KeywordRAG ($\mathcal{C}_{\rm{keyword}}$) and GraphRAG ($\mathcal{C}_{\rm{graph}}$), HybridRAG further incorporates the retrieval results from the VectorRAG module ($\mathcal{C}_{\rm{vector}}$)~\cite{lewis2020retrieval, su2024hybrid}, which is the traditional RAG based on vector databases. Finally, these outputs are merged to form the final retrieval result $\mathcal{C}_{\rm{final}}$, which is given by
\begin{equation}
\mathcal{C}_{\rm{final}} = \mathcal{C}_{\rm{keyword}} \cup \mathcal{C}_{\rm{graph}} \cup \mathcal{C}_{\rm{vector}}.
\end{equation}

The aggregated retrieval result $\mathcal{C}_{\rm{final}}$ with the original user query $\boldsymbol{Q}$ is then fed into LLM agents to generate comprehensive responses. The developed HybridRAG can enhance the retrieval quality by effectively integrating relevant information from different retrieval strategies. As a result, LLM agents can progressively generate domain-specific components of the carbon emission optimization problem.

\subsection{Problem Formulation}

After implementing the proposed framework, we leverage the HybridRAG-based LLM agent to formulate the carbon emission optimization problem for multi-UAV-assisted MEC networks through multiple interactions, as illustrated in Fig. \ref{Interaction_user_agent}. The generated optimization problem involves jointly optimizing task offloading scheduling $\boldsymbol{A} = \{\alpha_{k,m}(n),\:\forall k \in \mathcal{K}, m \in \mathcal{M}, n \in \mathcal{N}\}$, computing resource allocation $\boldsymbol{F} = \{f_{k,m}(n),\:\forall k \in \mathcal{K}, m \in \mathcal{M}, n \in \mathcal{N}\}$, and UAV trajectory control $\boldsymbol{V} = \{\mathbf{w}_m(n),\:\forall m \in \mathcal{M}, n \in \mathcal{N}\}$. We then simply organize the generated problem as needed, and this process is non-complicated~\cite{10815045, 10679152}. The final carbon emission optimization problem for multi-UAV-assisted MEC networks can be expressed as
\begin{subequations}\label{problem1}
    \begin{align}
        &\min\limits_{\{\boldsymbol{A}, \boldsymbol{F}, \boldsymbol{V}\}}\:C^{\rm{Total}} \\
        &\:\:\mathrm{s.t.}\:\:\alpha_{k,m}(n) \in \{0,1\},\label{constraint_offloading_1}\\
        &\qquad\: \sum\nolimits_{m=1}^M \alpha_{k,m}(n) = 1,\label{constraint_offloading_2}\\
        &\qquad\: x_m(n) \in [0, X_{\rm{max}}],\: y_m(n)\in [0, Y_{\rm{max}}],\label{constraint_mobility_1}\\
        &\qquad\: \lVert\mathbf{w}_m(n+1) - \mathbf{w}_m(n) \rVert \leq \delta_t V_{\rm{max}},\label{constraint_mobility_2}\\
        &\qquad\: \lVert \mathbf{w}_m(n) - \mathbf{w}_{m^{\prime}}(n) \rVert \geq D_{\rm{min}},\: \forall m, m^{\prime}, m \neq m^{\prime},\label{constraint_mobility_3}\\
        &\qquad\: \alpha_{k,m}(n)\lVert \mathbf{w}_m(n) - \mathbf{v}_k(n)\rVert^2 \leq r_{\rm{max}}^2 + H^2,\label{constraint_mobility_4}\\
        &\qquad\: \sum\nolimits_{k = 1}^K B_{k,m}^{\rm{G2A}} \leq B_{\rm{max}}^{\rm{G2A}},\label{constraint_bandwidth}\\
        &\qquad\: \sum\nolimits_{k=1}^K \alpha_{k,m}(n) f_{k,m}(n) \leq F_{m}^{\mathrm{max}},\label{constraint_computing_1}\\
        &\qquad\: 0 \leq f_{k,m}(n)\leq F_{m}^{\mathrm{max}}.\label{constraint_computing_2}
    \end{align}
\end{subequations}
Constraints (\ref{constraint_offloading_1}) and (\ref{constraint_offloading_2}) represent the task offloading constraints of users. Constraints (\ref{constraint_mobility_1})-(\ref{constraint_mobility_4}) depict the mobility restrictions imposed on UAVs. Constraint (\ref{constraint_bandwidth}) represents the uplink bandwidth constraint for G2A communications. Constraints (\ref{constraint_computing_1}) and (\ref{constraint_computing_2}) represent the computing resource restrictions of UAVs, which cannot exceed the maximum computing capacity of UAVs $F_{m}^{\mathrm{max}}$.

\begin{theorem}\label{theorem_NP_hard}
The optimization problem (\ref{problem1}) is a non-convex and NP-hard function.
\end{theorem}
\begin{proof}
    The objective variables involve binary variables (i.e., task offloading scheduling $\boldsymbol{A}$) and continuous variables (i.e., computing resource allocation $\boldsymbol{F}$ and UAV trajectory control $\boldsymbol{V}$). In addition, the constraint (\ref{constraint_offloading_1}) is non-convex. Thus, the optimization problem (\ref{problem1}) is a mixed-integer nonlinear programming NP-hard problem, which is non-convex.
\end{proof}

According to Theorem \ref{theorem_NP_hard}, it is computationally intractable to directly solve the optimization problem (\ref{problem1}) by using heuristic algorithms with finite time. Furthermore, the multi-UAV-assisted MEC network environment is highly dynamic, with variables such as user locations and channel conditions between users and UAVs constantly changing. To address these challenges, we adopt DRL algorithms, which enable model-free policy learning through data sampling to tackle the optimization problem (\ref{problem1}). However, traditional DRL algorithms often struggle with exploration in high-dimensional and complex network environments~\cite{HongyangTMC, WenIoTJ}, leading to convergence to suboptimal policies. Therefore, we introduce diffusion models as DRL policies due to their superior abilities to capture multi-dimensional network features. Building upon this, we propose the R\textsuperscript{2}DSAC algorithm to effectively learn optimal policies to solve the optimization problem (\ref{problem1}).

\begin{figure}
    \centering
    \includegraphics[width=0.4\textwidth]{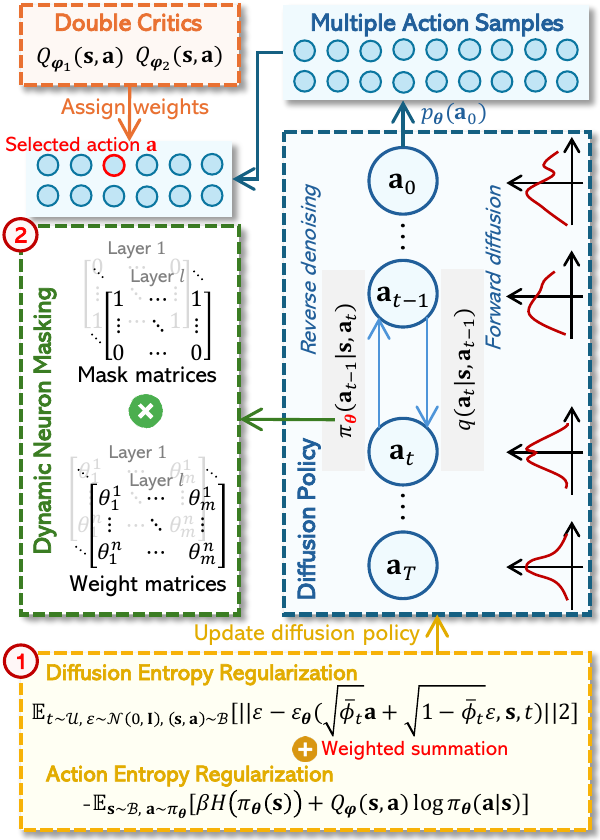}
    \caption{The architecture of the R\textsuperscript{2}DSAC algorithm, with two improvements to traditional diffusion-based DRL algorithms. The first improvement is that we incorporate diffusion entropy regularization and action entropy regularization into the diffusion policy, thereby enhancing policy performance. The second improvement is that we dynamically mask unimportant neurons of diffusion-based actor networks, thereby reducing carbon emissions due to model training.}
    \label{Archiecture_RDMSAC}
\end{figure}

\section{Double Regularization Diffusion-enhanced Soft Actor-Critic Algorithms}\label{RDMSAC}
In this section, we first model the optimization problem (\ref{problem1}) as a Markov Decision Process (MDP). Then, we present the architecture of the proposed R\textsuperscript{2}DSAC algorithm.

\subsection{MDP Modeling}
In the multi-UAV-assisted MEC network, UAVs are generally coordinated by a central UAV manager, which is responsible for both information management and UAV control~\cite{10134570}. Considering that the current actions, namely task offloading scheduling $\boldsymbol{A}$, computing resource allocation $\boldsymbol{F}$, and UAV trajectory control $\boldsymbol{V}$, taken by the UAV manager may affect the following environmental state~\cite{ZhaoTWC}, we formulate the optimization problem (\ref{problem1}) as a MDP $\langle \mathcal{S}, \mathcal{A}, \mathcal{P}, \mathcal{R}, \gamma \rangle$, where $\mathcal{S}$ is the state space, $\mathcal{A}$ is the action space of the DRL agent (i.e., the central UAV manager), $\mathcal{P}$ represents the state transition probability, $\mathcal{R}$ is the reward function of the DRL agent, and $\gamma \in [0,1]$ is the discount factor controlling future returns. The detailed designs are shown as follows:
\subsubsection{State space} At each time slot $n$, the UAV manager can receive task information from users and clearly know the positions of UAVs~\cite{10381761}. Thus, the state $\mathbf{s}(n)$ is composed of task information $\boldsymbol{I}_k(n)$ and UAV positions $\mathbf{w}_m(n) = [x_m(n), y_m(n), H]^T$, which is given by
\begin{equation}
   \mathbf{s}(n) \triangleq \{\mathbf{w}_m(n), \boldsymbol{I}_k(n),\:\forall k \in \mathcal{K}, m \in \mathcal{M}\},
\end{equation}
where the total dimensionality of the state $\mathbf{s}(n)$ is $(2K + 3M)$.

\subsubsection{Action space} The action of the DRL agent $\mathbf{a}(n)$ involves $\boldsymbol{A}$, $\boldsymbol{F}$, and $\boldsymbol{V}$. To reduce the dimensionality of $\mathbf{a}(n)$, we define $i_k(n) \in \{1,\ldots, m, \ldots, M\}$ to represent the task offloading destination of user $k$ at time slot $n$~\cite{10381761}, where $i_k(n) = m$ indicates that user $k$ offloads its task to UAV $m$. Thus, the action $\mathbf{a}(n)$ at time slot $n$ is given by
\begin{equation}
    \mathbf{a}(n) \triangleq \{i_k(n), f_{k,m}(n), v_m(n), \theta_m(n),\:\forall k \in \mathcal{K}, m \in \mathcal{M}\},
\end{equation}
where the total dimensionality of the action $\mathbf{a}(n)$ is $(2M + K(M + 1))$. Since $i_k(n)$ is a discrete variable, we convert it into a continuous representation using a uniform segmentation manner, thus mitigating instability in policy learning caused by the hybrid action space.

\subsubsection{Reward function} The reward function typically accounts for both the objective and the associated constraints of the optimization problem. We denote $C(n)$ as the carbon emissions of the multi-UAV-assisted MEC network at time slot $n$, where $\sum\nolimits_{n=1}^N C(n) = C^{\rm{Total}}$. According to~\cite{ZhaoTWC, 10381761, wang2024multiuavenabledmecnetworks}, the reward function $\mathcal{R}(\mathbf{a}(n)|\mathbf{s}(n))$ can be given by
\begin{equation}\label{reward_function}
    \mathcal{R}(\mathbf{a}(n)|\mathbf{s}(n)) = -C(n) + \underbrace{\Omega_d + \Omega_f + \Omega_g + \Omega_i}_{\rm{Constraint}\:\rm{terms}},
\end{equation}
where $\Omega_d$, $\Omega_f$, $\Omega_g$, and $\Omega_i$ are designed through reward shaping based on human knowledge~\cite{wentao}, thereby guiding the action of the DRL agent $\mathbf{a}(n)$ to satisfy the constraints (\ref{constraint_mobility_1}), (\ref{constraint_mobility_3}), (\ref{constraint_mobility_4}), and (\ref{constraint_computing_1}). Notably, other constraints can be guaranteed through linearly mapping original actions~\cite{10679152}.

\subsection{Algorithm Architecture}
\subsubsection{Diffusion policy}
As illustrated in Fig. \ref{Archiecture_RDMSAC}, at each time slot $n$, conditioned on the observed environment state $\mathbf{s}$, a diffusion policy $\pi_{\boldsymbol{\theta}}(\mathbf{s})$ generates multiple action samples from an action probability distribution $p_{\boldsymbol{\theta}}(\mathbf{a}_0)$ through two sequential Markov processes: forward diffusion and reverse denoising. We consider that the forward diffusion process consists of $T$ steps, denoted as $\mathcal{U} = \{1,\ldots,t,\ldots,T\}$. During the forward diffusion process, the Gaussian noise is gradually added to the target action $\mathbf{a}_0$ across $T$ steps, and we can obtain latent actions $\mathbf{a}_1,\mathbf{a}_2,\ldots,\mathbf{a}_T$. Owing to the Markov chain property~\cite{ding2024diffusionbasedreinforcementlearningqweighted}, the Gaussian noise sample $\mathbf{a}_T \sim \mathcal{N}(0,\mathbf{I})$ can be obtained by cumulatively multiplying the transition from $\mathbf{a}_{t-1}$ to $\mathbf{a}_t$, which is given by~\cite{HongyangTMC, wang2023diffusionpoliciesexpressivepolicy}
\begin{equation}
    q(\mathbf{a}_T|\mathbf{a}_0) = \prod_{t=1}^T\mathcal{N}(\mathbf{a}_t; \sqrt{1 - \psi_t} \mathbf{a}_{t-1}, \psi_t\mathbf{I}).
\end{equation}
Here, $\mathbf{I}$ denotes the identity matrix, and $\psi_t$ represents the noise variance controlled by the variational posterior scheduler at step $t$, which can be calculated  as~\cite{HongyangTMC, ding2024diffusionbasedreinforcementlearningqweighted}
\begin{equation}
\psi_t = 1 - e^{-\frac{\psi_{\mathrm{min}}}{T}-\frac{2t-1}{2T^2}(\psi_{\mathrm{max}} - \psi_{\mathrm{min}})},
\end{equation}
where $\psi_{\mathrm{min}}$ and $\psi_{\mathrm{max}}$ are constant parameters. 

In the reverse denoising process, $\mathbf{a}_0$ is progressively reconstructed from the noise sample $\mathbf{a}_T$ through a step-by-step denoising procedure, and the transition from $\mathbf{a}_t$ to $\mathbf{a}_{t-1}$ follows a Gaussian distribution~\cite{HongyangTMC, ding2024diffusionbasedreinforcementlearningqweighted}, which is given by
\begin{equation}
    p_{\boldsymbol{\theta}}(\mathbf{a}_{t-1}|\mathbf{a}_t) = \mathcal{N}(\mathbf{a}_{t-1}; \mu_{\boldsymbol{\theta}}(\mathbf{a}_t, \mathbf{s}, t), \Sigma_{\boldsymbol{\theta}}(\mathbf{a}_t, t)),
\end{equation}
where
\begin{equation}\label{mean}
    \mu_{\boldsymbol{\theta}}(\mathbf{a}_t, \mathbf{s}, t) = \frac{1}{\sqrt{\phi_t}}\bigg(\mathbf{a}_t - \frac{\psi_t\tanh(\boldsymbol{\varepsilon}_{\boldsymbol{\theta}}(\mathbf{a}_t, \mathbf{s}, t))}{\sqrt{1 - \bar{\phi}_t}}\bigg),
\end{equation}
\begin{equation}\label{covariance}
    \Sigma_{\boldsymbol{\theta}}(\mathbf{a}_t, t) = \frac{\psi_t (1 - \bar{\phi}_{t-1})}{1 - \bar{\phi}_t}\mathbf{I}.
\end{equation}
Here, $\phi_t = 1 - \psi_t$, $\bar{\phi}_{t} = \prod_{i=1}^t \phi_i$, and $\boldsymbol{\varepsilon}_{\boldsymbol{\theta}}(\mathbf{a}_t, \mathbf{s}, t)$ is a deep network parameterized by $\boldsymbol{\theta}$, which generates denoising noises conditioned on the state $\mathbf{s}$ and the current denoising step $t$. Therefore, the generative action distribution $p_{\boldsymbol{\theta}}(\mathbf{a}_0)$ conditioned on $\prod_{t=1}^T(1 - \psi_t) \approx 0$ is given by
\begin{equation}\label{final_action}
    p_{\boldsymbol{\theta}}(\mathbf{a}_0) = \mathcal{N}(\mathbf{a}_T; 0, \mathbf{I})\prod_{t=1}^T p_{\boldsymbol{\theta}}(\mathbf{a}_{t-1}|\mathbf{a}_t).
\end{equation}

Finally, the resulting action $\mathbf{a}_0$ can be sampled from the learned generative distribution $p_{\boldsymbol{\theta}}(\mathbf{a}_0)$, representing the most probable choice among multiple candidate actions~\cite{ding2024diffusionbasedreinforcementlearningqweighted}. We then apply a linear mapping to convert $\mathbf{a}_0$ into $\mathbf{a}$, which can be directly executed in the environment.  

\subsubsection{Q-learning guidance}
To facilitate the diffusion policy $\pi_{\boldsymbol{\theta}}(\mathbf{s})$ to generate actions that contribute to minimizing carbon emissions in the multi-UAV-assisted MEC network, we employ Q-learning guidance into the learning of $\boldsymbol{\varepsilon}_{\boldsymbol{\theta}}(\mathbf{a}_t, \mathbf{s}, t)$ during the reverse denoising process, enabling $\pi_{\boldsymbol{\theta}}(\mathbf{s})$ to sample actions with high $Q$ values. Specifically, we first construct two Q-networks $Q_{\boldsymbol{\varphi}_1}, Q_{\boldsymbol{\varphi}_2}$ and two target networks $Q_{\hat{\boldsymbol{\varphi}}_1}, Q_{\hat{\boldsymbol{\varphi}}_2}$. These critic networks possess the same network architecture. To optimize $\boldsymbol{\varphi}_i$ for $i = {1, 2}$, we minimize the temporal difference error, which is expressed as
\begin{equation}\label{Q_update}
    \begin{split}
    &\mathbb{E}_{(\mathbf{s}(n), \mathbf{a}(n), \mathbf{s}(n+1), \mathcal{R})\sim \mathcal{O}}\Big[\sum_{i=1,2}(\mathcal{R}(\mathbf{a}(n)|\mathbf{s}(n))+\gamma^n(1-d_{n+1})\\
    &(Q_{\hat{\boldsymbol{\varphi}}}(\mathbf{s}(n+1)) - \beta \log \pi_{\hat{\boldsymbol{\theta}}}(\mathbf{s}(n+1)))-Q_{\boldsymbol{\varphi}_i}(\mathbf{s}(n),\mathbf{a}(n)))^2\Big],
\end{split}
\end{equation}
where $\mathcal{O}$ is a mini-batch of transitions sampled from a relay buffer $\mathcal{B}$, $Q_{\hat{\boldsymbol{\varphi}}}(\mathbf{s}) = \min \{Q_{\hat{\boldsymbol{\varphi}}_1}(\mathbf{s}), Q_{\hat{\boldsymbol{\varphi}}_2}(\mathbf{s})\}$~\cite{HongyangTMC}, $d_{n+1} \in \{0,1\}$ is a terminated flag~\cite{HongyangTMC}, with $d_{n+1} = 1$ meaning that the training episode (i.e., $N$ time slots) is ended, $\beta$ is a temperature coefficient that controls the trade-off between the entropy term and the reward~\cite{pmlr-v80-haarnoja18b}, and $\pi_{\hat{\boldsymbol{\theta}}}(\mathbf{s})$ is the target diffusion policy.

\subsubsection{Policy improvement module}
To optimize the diffusion policy $\pi_{\boldsymbol{\theta}}(\mathbf{s})$, instead of directly optimizing the state-value function $V^{\pi}(\mathbf{s})$~\cite{kang2025confidenceregulatedgenerativediffusionmodels}, two issues need to be resolved for the practical implementation of the R\textsuperscript{2}DSAC algorithm in multi-UAV-assisted MEC networks: 
\begin{itemize}
    \item \textit{Negative $Q$ values:} In the reward function (\ref{reward_function}), since $-C(n)$ is inherently negative, it becomes difficult to guarantee that the returned reward always remains non-negative, resulting in negative $Q$ values for certain state-action pairs~\cite{ding2024diffusionbasedreinforcementlearningqweighted}, which may lead to the instability of diffusion policy learning.
    \item \textit{Limited high-quality training samples:} In multi-UAV-assisted MEC networks, it is challenging to obtain expert datasets consisting of state-action samples with high $Q$ values to effectively guide policy learning~\cite{ding2024diffusionbasedreinforcementlearningqweighted}. In the absence of expert behavior guidance, the diffusion policy may be overly confident in specific actions, potentially leading to convergence to a suboptimal solution~\cite{HongyangTMC}.
\end{itemize}

To address the above issues, we incorporate action entropy regularization and diffusion entropy regularization into the policy learning objective function. Specifically, we introduce an action entropy regularization term to encourage the policy to generate a more uniform action distribution, as given by
\begin{equation}\label{action_entropy}
\begin{aligned}
    \mathcal{L}_{\mathrm{act}}(\boldsymbol{\theta}) = - \mathbb{E}_{\mathbf{s}\sim\mathcal{B},\mathbf{a}\sim \pi_{\boldsymbol{\theta}}} [\beta H(\pi_{\boldsymbol{\theta}}(\mathbf{s})) + Q_{\boldsymbol{\varphi}}(\mathbf{s}, \mathbf{a}) \log\pi_{\boldsymbol{\theta}}(\mathbf{a}|\mathbf{s})],
\end{aligned}
\end{equation}
where $Q_{\boldsymbol{\varphi}}(\mathbf{s}, \mathbf{a}) = \min \{Q_{\boldsymbol{\varphi}_1}(\mathbf{s}, \mathbf{a}), Q_{\boldsymbol{\varphi}_2}(\mathbf{s}, \mathbf{a})\}$ and $H(\pi_{\boldsymbol{\theta}}(\mathbf{s}))$ denotes the entropy of the action probability distribution~\cite{HongyangTMC, pmlr-v80-haarnoja18b}. It is worth noting that the action entropy regularization in (\ref{action_entropy}) serves to prevent the policy from prematurely converging to a suboptimal solution~\cite{HongyangTMC}.

The diffusion policy regularization is a sampling-based approach that requires only the random sampling of state-action pairs $(\mathbf{s}(n), \mathbf{a}(n))$ from the relay buffer $\mathcal{B}$ and the current policy~\cite{wang2023diffusionpoliciesexpressivepolicy}. Inspired by the process of denoising diffusion probabilistic models in image generation, we utilize the mean squared error loss to represent the diffusion policy regularization~\cite{wang2023diffusionpoliciesexpressivepolicy}, which is given by
\begin{equation}\label{diffusion_entropy}
\begin{aligned}
    \mathcal{L}_{\mathrm{diff}}(\boldsymbol{\theta}) = \mathbb{E} \Bigg[\bigg|\bigg|\boldsymbol{\varepsilon} - \boldsymbol{\varepsilon}_{\boldsymbol{\theta}}\bigg(\sqrt{\bar{\phi}_{t}}\mathbf{a} + \sqrt{1 - \bar{\phi}_{t}}\boldsymbol{\varepsilon},\mathbf{s}, t\bigg)\bigg|\bigg|^2\Bigg],
\end{aligned}
\end{equation}
where $\boldsymbol{\varepsilon} \sim \mathcal{N}(0, \mathbf{I})$ and $\sqrt{\bar{\phi}_{t}}\mathbf{a} + \sqrt{1 - \bar{\phi}_{t}}\boldsymbol{\varepsilon}$ represents the expert action after the reverse denoising process. Notably, $\mathcal{L}_{\mathrm{diff}}(\boldsymbol{\theta})$ can be seen as a behavior-cloning loss.

Therefore, we formulate the final policy-learning objective as a weighted combination of the diffusion policy regularization and the action policy regularization, which is given by~\cite{wang2023diffusionpoliciesexpressivepolicy}
\begin{equation}
    \pi = \argmin_{\pi_{\boldsymbol{\theta}}}(\mathcal{L}(\boldsymbol{\theta}) = \rho \mathcal{L}_{\mathrm{act}}(\boldsymbol{\theta}) + (1 - \rho) \mathcal{L}_{\mathrm{diff}}(\boldsymbol{\theta})),
\end{equation}
where $\rho \in [0,1]$ represents a behavior-cloning weight.

\subsubsection{Dynamic pruning module}
The actor model $\boldsymbol{\varepsilon}_{\boldsymbol{\theta}}(\mathbf{a}_t, \mathbf{s}, t)$ incorporates temporal information by using sinusoidal position embeddings with multiple Fully Connected Layers (FCLs)~\cite{HongyangTMC}. The encoded time vector is concatenated with the state $\mathbf{s}$ and the noise sample $\mathbf{a}_T$ and passed through a Multi-Layer Perceptron (MLP). Finally, this MLP maps the concatenated input to an output action $\mathbf{a}_{0}$ bounded by a $\tanh$ activation. To reduce carbon emissions associated with model training, we design a dynamic pruning module to dynamically suppress the activity of unimportant neurons in FCLs of $\boldsymbol{\varepsilon}_{\boldsymbol{\theta}}(\mathbf{a}_t, \mathbf{s}, t)$~\cite{10818642}. Specifically, at the beginning of each train episode, the pruning module first evaluates the importance of neurons in each FCL, and then it masks the top $\lfloor |\boldsymbol{\theta}^{(\ell)}| \cdot \varrho \rfloor$ neurons with the lowest importance scores, indicating that their corresponding weight vectors are assigned to $0$, where $|\boldsymbol{\theta}^{(\ell)}|$ denotes the total number of neurons in the FCL $\ell$ and  $\varrho \in [0,1)$ is the pruning rate. This mathematical process can be expressed as
\begin{equation}
    \boldsymbol{\theta}^{(\ell)} \gets \boldsymbol{\theta}^{(\ell)} \odot \boldsymbol{M}^{(\ell)},
\end{equation}
where $\odot$ is the Hadamard product, and $\boldsymbol{M}^{(\ell)}$ represent the mask matrix for the FCL $\ell$, with values of either $0$ or $1$. Similarly, the unimportant parameters of the target actor network are also masked by the dynamic pruning module. 

\subsubsection{Parameter updates} At the end of the training episode, we use the Adam optimizer to update the policy parameters $\boldsymbol{\theta}$, which is given by
\begin{equation}\label{actor_update}
    \boldsymbol{\theta}_{e+1} \gets \boldsymbol{\theta}_{e} - \sigma \nabla_{\boldsymbol{\theta}}\mathcal{L}(\boldsymbol{\theta}),
\end{equation}
where $\boldsymbol{\theta}_{e}$ are the policy parameters in the $e$-th training episode and $\sigma \in (0,1]$ is the learning rate of the policy. In addition, we perform a soft update to the parameters of the target actor and critic networks~\cite{wang2023diffusionpoliciesexpressivepolicy}, respectively, as given by
\begin{equation}\label{target_update}
\begin{split}
    \hat{\boldsymbol{\theta}}_{e+1} &\gets \xi \boldsymbol{\theta}_e + (1 - \xi)\hat{\boldsymbol{\theta}}_{e},\\
    \hat{\boldsymbol{\varphi}}_{e+1} &\gets \xi \boldsymbol{\varphi}_{e} + (1-\xi)\hat{\boldsymbol{\varphi}}_{e},
\end{split}  
\end{equation}
where $\boldsymbol{\varphi}_{e}$ are the Q-function parameters $\boldsymbol{\varphi}_{i, e},\:i\in \{1,2\}$ in the $e$-th training episode and $\xi \in (0,1]$ is the update rate of target networks.

\begin{algorithm}[t]
\label{diffusion_algorithm}
\DontPrintSemicolon
\SetAlgoLined

Initialize parameters $\boldsymbol{\theta}, \boldsymbol{\varphi}, \hat{\boldsymbol{\theta}}, \hat{\boldsymbol{\varphi}}$ and relay buffer $\mathcal{B}$.

Initialize mask matrices $\boldsymbol{M}$.

Initialize hyperparameters and pruning rate $\varrho$.

\For{\rm{the episode} $e=1$ \rm{to} $E_{\mathrm{max}}$}
{
\For{$n=1$ \rm{to} $N$}
{
    Observe state $\mathbf{s}(n)$ and randomly initialize a normal sample $\mathbf{a}_T \sim \mathcal{N}(0, \mathbf{I})$.
    
    \For{$t=1$ \rm{to} $T$}
    {
    \textcolor{blue}{\textit{\#\#\# Reverse denoising \#\#\#}}
    
    Construct a denoising network $\boldsymbol{\varepsilon}_{\boldsymbol{\theta}}(\mathbf{a}_t, \mathbf{s}, t)$.

    Calculate the mean and covariance by using (\ref{mean}) and (\ref{covariance}), respectively.
    
    Obtain the action $\boldsymbol{a}_0$ based on (\ref{final_action}). 
    }

    \textcolor{blue}{\textit{\#\#\# Experience collections \#\#\#}}
    
    Linearly map $\boldsymbol{a}_0$ to $\boldsymbol{a}(n)$ and perform $\boldsymbol{a}(n)$.
    
    Observe the next state $\mathbf{s}(n+1)$ and obtain the corresponding reward $\mathcal{R}(\mathbf{a}(n)|\mathbf{s}(n))$. 

    Store record $(\mathbf{s}(n),\mathbf{a}(n),\mathbf{s}(n+1), \mathcal{R})$ into $\mathcal{B}$.

}

    \textcolor{blue}{\textit{\#\#\# Dynamic pruning \#\#\#}}
    
    Evaluate the importance of neurons in $\boldsymbol{\theta}$ and $\hat{\boldsymbol{\theta}}$.

    Dynamically mask the unimportant neurons of $\boldsymbol{\theta}$ and $\hat{\boldsymbol{\theta}}$ according to the pruning rate $\varrho$.

    \textcolor{blue}{\textit{\#\#\# Parameter updates \#\#\#}}
    
    Sample a random mini-batch of transitions $\mathcal{O}$ with size $O$ from $\mathcal{B}$.
    
    Update $Q_{\boldsymbol{\varphi}_1}, Q_{\boldsymbol{\varphi}_2}$ using $\mathcal{B}$ to minimize (\ref{Q_update}).
    
    Update the policy parameters $\boldsymbol{\theta}$ using $\mathcal{B}$ by (\ref{actor_update}).

    Update target network parameters $\hat{\boldsymbol{\theta}}, \hat{\boldsymbol{\varphi}}$ by (\ref{target_update}).
}
\textbf{return} the policy networks.

\caption{R\textsuperscript{2}DSAC}\label{diffusion_algorithm}
\end{algorithm}

\subsection{Complexity Analysis}
Algorithm \ref{diffusion_algorithm} presents the comprehensive process for implementing the R\textsuperscript{2}DSAC algorithm. In the following, we analyze its computational complexity. 
\subsubsection{Algorithm initialization} 
The computational overhead of algorithm initialization mainly comes from the initialization of network parameters and mask matrices, and the computational complexity of this part is $\mathcal{O}(4|\boldsymbol{\theta}| + 2|\boldsymbol{\varphi}|)$.

\subsubsection{Action sampling}
The computational complexity of action sampling arises from the reverse diffusion process, which is given by $\mathcal{O}(E_{\mathrm{max}}NT|\boldsymbol{\theta}|)$~\cite{HongyangTMC}.

\subsubsection{Experience collections}
We define the complexity of the DRL agent interacting with the environment as $V$. The computational complexity of experience collections is $\mathcal{O}(E_{\mathrm{max}}NV)$~\cite{HongyangTMC}. 

\subsubsection{Dynamic pruning}
In the dynamic pruning module, the computational overhead comes from the evaluation of neuron importance and the Hadamard product. Hence, the computational complexity of dynamic pruning is $\mathcal{O}(2E_{\mathrm{max}}|\boldsymbol{\theta}|)$~\cite{10818642}.

\subsubsection{Parameter updates}
The computational complexity of parameter updates consists of three parts: $\mathcal{O}(OE_{\mathrm{max}}|\boldsymbol{\theta}|)$ for policy improvement, $\mathcal{O}(OE_{\mathrm{max}}|\boldsymbol{\varphi}|)$ for critic network improvement, and $\mathcal{O}(E_{\mathrm{max}}(|\boldsymbol{\theta}| + |\boldsymbol{\varphi}|))$ for target network improvement. Thus, the computational complexity of parameter updates is $\mathcal{O}(E_{\mathrm{max}}(O+1)(|\boldsymbol{\theta}| + |\boldsymbol{\varphi}|))$~\cite{HongyangTMC}.

Based on the above analysis, the computational complexity of the R\textsuperscript{2}DSAC algorithm is $\mathcal{O}(4|\boldsymbol{\theta}| + 2|\boldsymbol{\varphi}| + E_{\mathrm{max}}N(T|\boldsymbol{\theta}| + V) + 2E_{\mathrm{max}}|\boldsymbol{\theta}| +  E_{\mathrm{max}}(O+1)(|\boldsymbol{\theta}| + |\boldsymbol{\varphi}|))$.

\begin{table}[t]\label{parameter}
	\renewcommand{\arraystretch}{1.2}
	\caption{Simulation Parameters}
    \centering
	\begin{tabular}{m{5.3cm}<{\raggedright}|m{2.5cm}<{\centering}}
    \toprule
		\hline		
		\textbf{Parameters} & \textbf{Values}\\	
		\hline
        Data sizes of computing tasks ($D_k(n)$)~\cite{9817088} &  $[100, 300]\: \rm{MB}$\\	
        \hline
        Number of CPU cycles required for computing one bit of task data ($C_k(n)$)~\cite{ZhaoTWC} &  $[100, 200]\: \rm{cycles/bit}$\\	
		\hline
        Fixed flight altitude of UAVs ($H$)~\cite{9930881} &  $100\:\rm{m}$\\	
		\hline	
		Maximum speed of UAVs ($V_{\rm{max}}$)~\cite{SunTMC}  &  $60\:\rm{m/s}$  \\	
		\hline
		Safe distance between UAVs ($D_{\rm{min}}$)~\cite{9930881}   & $10\:\rm{m}$ \\
		\hline	
		  Coverage area radius of UAVs ($r_{\rm{max}}$)~\cite{9930881} &  $100\:\rm{m}$\\
		\hline
        Constant parameters ($a$, $b$)~\cite{wang2024multiuavenabledmecnetworks} &  $9.61$, $0.16$\\
		\hline
        Carrier frequency of UAVs ($f$)~\cite{9817088} &  $2\:\rm{GHz}$\\ \hline
        Excessive path loss for LoS and NLoS links ($\eta^{\rm{LoS}}$, $\eta^{\rm{NLoS}}$)~\cite{10381761}&  $1$, $20$\\	
		\hline
		  Maximum uplink bandwidth ($B_{\rm{max}}^{\rm{G2A}}$)~\cite{SunTMC}  &  $20\:\rm{MHz}$  \\	
		\hline
		Transmit power of user $k$ ($p_k(n)$)~\cite{9817088}   & $23\:\rm{dBm}$  \\
		\hline
		Additive Gaussian white noise power ($\delta^2$)~\cite{SunTMC} &  $-100\:\rm{dBm}$\\	
		\hline		
		  Effective switching capacitance ($\epsilon_m$)~\cite{SunTMC} &  $10^{-27}$\\
		\hline
        Computing capacity of UAV $m$ ($F_{m}^{\mathrm{max}}$)~\cite{9930881} &  $5\:\rm{GHz}$\\
		\hline
        Blade profile power ($P_0$)~\cite{9930881} &  $79.8563\:\rm{Watt}$\\ \hline
        Induced power for hovering ($P_1$)~\cite{9930881} &  $88.6279\:\rm{Watt}$\\ \hline
        Blade tip speed ($U^{\rm{Tip}}$)~\cite{9930881} &  $120\:\rm{m/s}$\\ \hline
        UAV fuselage drag ratio ($d_0$)~\cite{9930881} &  $0.6$\\ \hline
        Air density ($\zeta$)~\cite{9930881} &  $1.225\:\rm{kg/m^3}$\\ \hline
        Rotor solidity ($s$)~\cite{9930881} &  $0.05$\\ \hline
        Rotor disk area ($A$)~\cite{9930881} &  $0.503\:\rm{m^2}$\\ \hline
        Average rotor induced speed ($v_0$)~\cite{9930881} &  $4.03\:\rm{m/s}$\\ \hline
        Conversion coefficient between watt-hours and joules ($\tau$)~\cite{10677514} &  $1/3600$\\ \hline
    \bottomrule
	\end{tabular}\label{parameter}
\end{table}

\begin{figure}
    \centering
    \includegraphics[width=0.40\textwidth]{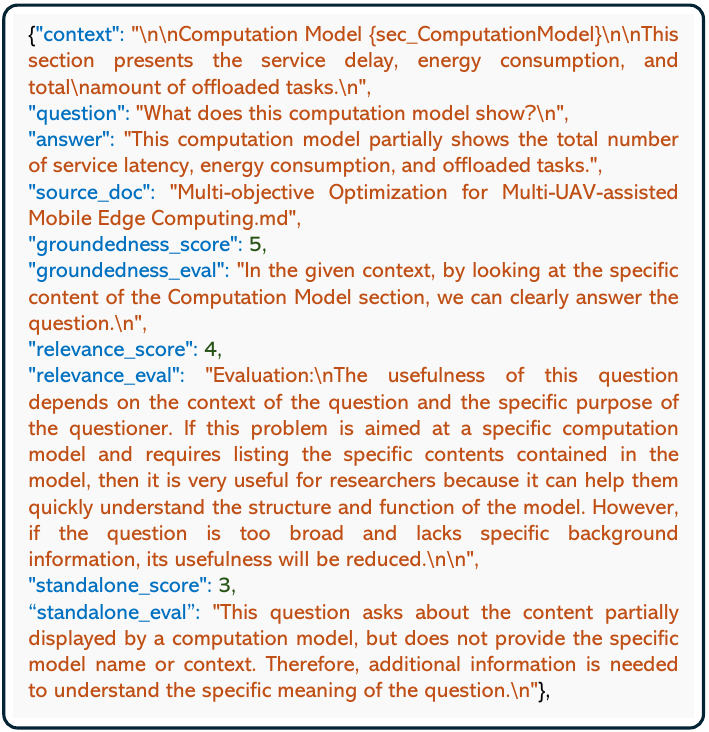}
    \caption{QA pairs outputted by the HybridRAG-based LLM agent, which can serve as the test dataset to evaluate the performance of HybridRAG.}
    \label{Critique_qa}
\end{figure}

\begin{table}[htbp]
  \renewcommand{\arraystretch}{1.2}
  \caption{Performance Evaluation Results of HybridRAG}
  \centering
  \begin{tabular}{ccc}
    \toprule
    \multirow{2}{*}{\textbf{Metrics}} & \multicolumn{2}{c}{\textbf{RAG Systems}} \\
    \cmidrule(lr){2-3}
    & \textbf{VectorRAG + KeywordRAG} & \textbf{HybridRAG} \\
    \midrule
    Prec.$\uparrow$ & $44.6$ & \cellcolor{gray!10}$\textbf{46.2}$ \\
    Rec.$\uparrow$ & $74.4$ & \cellcolor{gray!10}$\textbf{76.5}$ \\
    F1$\uparrow$ & $49.9$ & \cellcolor{gray!10}$\textbf{53.2}$ \\ 
    \midrule[0.1pt]
    \addlinespace[0.5em] 
    CR$\uparrow$ & $82.7$ & \cellcolor{gray!10}$\textbf{83.1}$ \\
    CP$\uparrow$ & $32.4$ & \cellcolor{gray!10}$\textbf{33.4}$ \\ 
    \midrule[0.1pt]
    \addlinespace[0.5em] 
    CU$\uparrow$ & $75.8$ & \cellcolor{gray!10}$\textbf{80.2}$ \\
    NS(I)$\downarrow$ & $30.5$ & \cellcolor{gray!10}$\textbf{28.6}$ \\
    NS(II)$\downarrow$ & $18.1$ & \cellcolor{gray!10}$\textbf{18}$ \\ 
    Hallu.$\downarrow$ & $\textbf{6.7}$ & \cellcolor{gray!10}$7.2$ \\
    SK$\downarrow$ & $0.4$ & \cellcolor{gray!10}$0.4$ \\
    Faith.$\uparrow$ & $\textbf{92.8}$ & \cellcolor{gray!10}$92.4$ \\
    \bottomrule
  \end{tabular}
  \label{RAGChecker}
\end{table}

\section{Simulation Results}\label{Simulation_result}
In this section, we first introduce the experimental setup. We then employ the RAGChecker framework to evaluate the performance of the developed HybridRAG. Finally, we validate the effectiveness of the proposed R\textsuperscript{2}DSAC algorithm.

\subsection{Experimental Setup}
We consider a multi-UAV-assisted MEC network where $2$ UAVs possess offloaded tasks and provide services to $10$ users in a $1000 \times 1000 \: \rm{m}^2$ rectangular area, i.e., $X_{\rm{max}} = Y_{\rm{max}} = 1000$. The service duration is configured as $N = 100$ time slots,  with $\delta_t = 1\:\rm{s}$~\cite{9930881}. Each UAV is equipped with an MEC server and is capable of serving up to $5$ users at each time slot~\cite{10636964}. Without loss of generality, the initial positions of UAVs are set to $(400, 400, 100)$ and $(600, 600, 100)$, respectively. Table \ref{parameter} shows the simulation parameters, and the experiments for the performance evaluation of the R\textsuperscript{2}DSAC algorithm are conducted on an NVIDIA GeForce RTX 3080 Laptop GPU by using PyTorch with CUDA 12.0.

For the construction of HybridRAG, we call the Qwen2.5-72B model through API as the pluggable LLM module, with the temperature set to $0.85$ and the context window configured to $8192$ tokens. Moreover, we utilize the BGE-M3 Embedding model\footnote{\url{https://github.com/FlagOpen/FlagEmbedding}} to transform textual data into high-dimensional vector representations, and the chunk size is set to $1024$, with a default chunk overlap of $20$ tokens. Furthermore, we employ Qdrant to store the high-dimensional vectors, utilize Neo4j to manage the knowledge graph, and adopt the Simple Keyword Table Index for keyword extraction from the text. The experiments for the performance evaluation of HybridRAG are conducted on an Intel Xeon(R) Gold 6133 CPU and two NVIDIA RTX A6000 GPUs.

\begin{figure*}[t]
\centering
\subfigure[Comparisons of test rewards for different algorithms.]
{
    \begin{minipage}[t]{0.47\linewidth}
	\centering
	\includegraphics[width=1\linewidth]{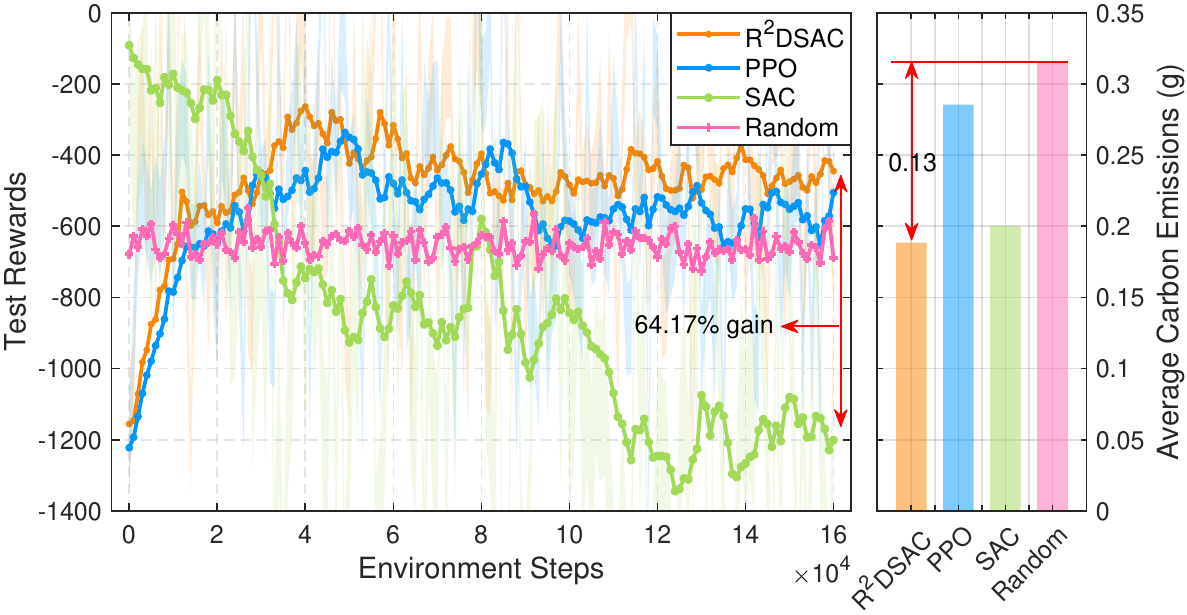}\label{Algorithm_compare}
        \captionsetup{font=footnotesize}
    \end{minipage}
}
\hspace{0.05in}
\subfigure[Ablation experiments.]
{
    \begin{minipage}[t]{0.47\linewidth}
	\centering
	\includegraphics[width=1\linewidth]{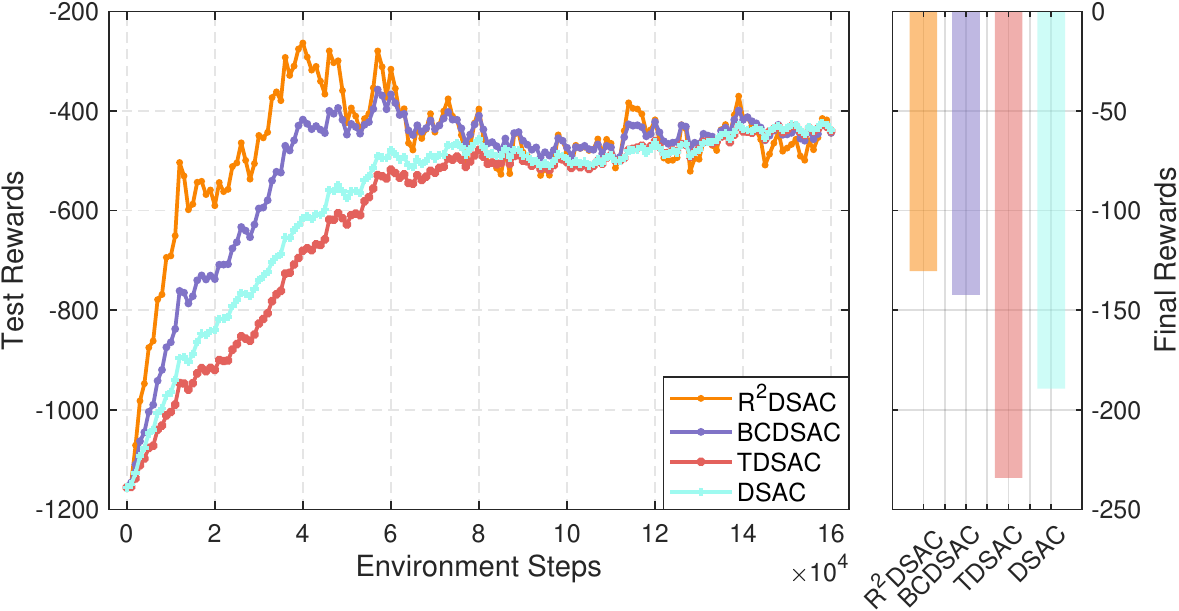}\label{ablation}
        \captionsetup{font=footnotesize}
    \end{minipage}
}
\caption{Performance evaluation of the R\textsuperscript{2}DSAC algorithm in carbon emission optimization.}\label{Performance_compare}
\end{figure*}

\begin{figure*}[t]
		\centering
		\begin{subfigure}[Pruning rate impact.]{\includegraphics[width=0.3\linewidth]{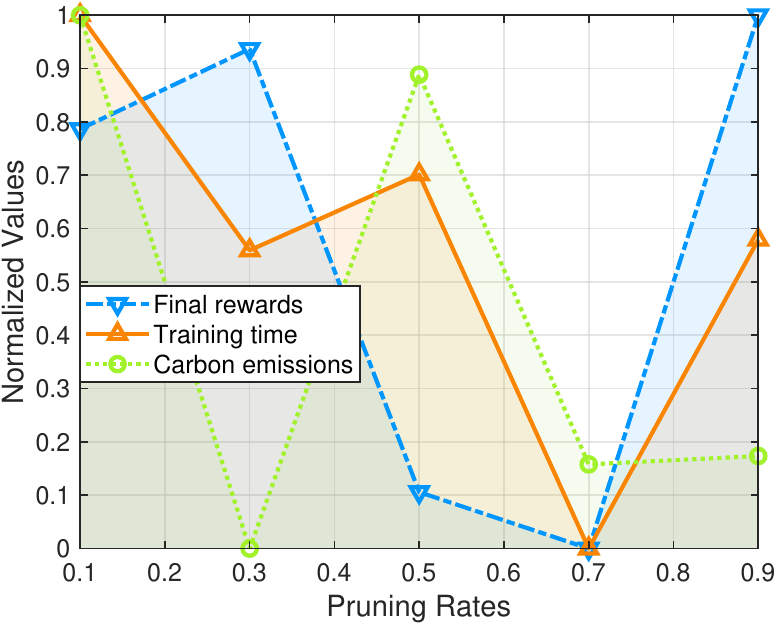}\label{Pruning_normalization}}
        \captionsetup{font=footnotesize}
		\end{subfigure}
		\hfill
		\begin{subfigure}[Diffusion step impact.]{\includegraphics[width=0.3\linewidth]{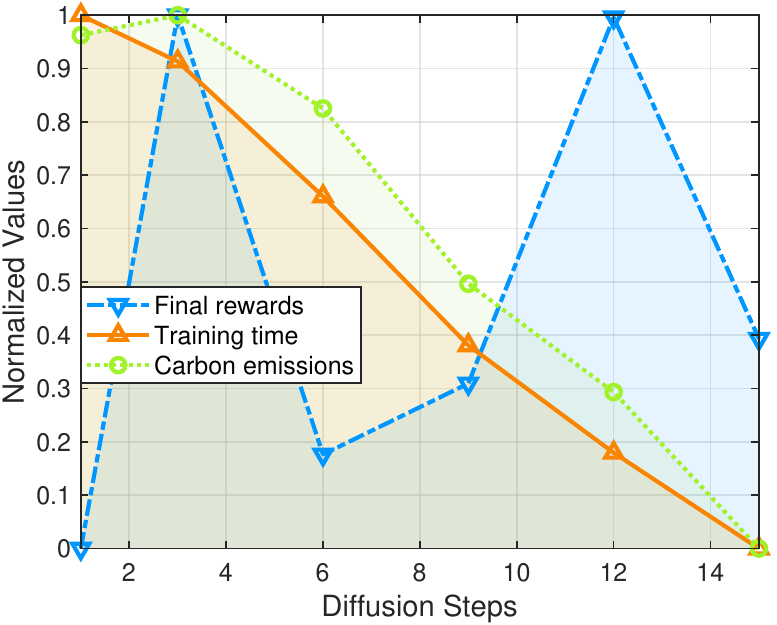}\label{Diffusion_normalization}}
        \captionsetup{font=footnotesize}
		\end{subfigure}
		\hfill
		\begin{subfigure}[Behavior-cloning weight impact.]{\includegraphics[width=0.3\linewidth]{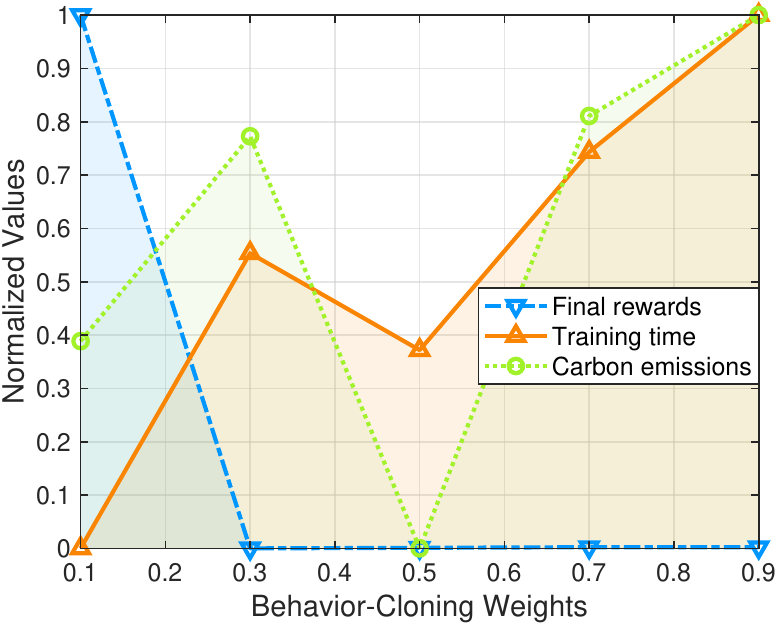}\label{Behavior_normalization}}
        \captionsetup{font=footnotesize}
		\end{subfigure}
		\caption{Impacts of pre-defined parameters on final rewards, training time, and carbon emissions during model training.}
		\label{Normalization_analysis}
\end{figure*}

\subsection{Performance Evaluation of HybridRAG}
We apply RAGChecker to evaluate the performance of HybridRAG~\cite{NEURIPS2024_27245589}. Prior to evaluation, we construct a baseline test dataset consisting of Question-Answer (QA) pairs generated by the HybridRAG-based LLM agent, as shown in Fig. \ref{Critique_qa}. The RAGChecker adopts a claim-level checking method to enable fine-grained performance evaluation. It leverages LLMs to perform two roles: extracting claims from the test dataset as text-to-claim extractors and verifying their accuracy as claim-entailment checkers~\cite{NEURIPS2024_27245589}. In our implementation, we utilize Qwen2.5-14B models as extractors and checkers. To evaluate the performance of HybridRAG, we adopt three categories of evaluation metrics: overall metrics, generator metrics, and retriever metrics~\cite{NEURIPS2024_27245589}, as illustrated in Table \ref{RAGChecker}.

\subsubsection{Overall metrics} To assess the overall response quality of HybridRAG, we compute claim-level Precision (Prec.) and Recall (Rec.), where Prec. measures the proportion of correct claims among all response claims, while Rec. quantifies the proportion of correct claims in all ground-truth answer claims. Moreover, we compute the F1 score as the overall performance metric by calculating the harmonic mean of Prec. and Rec. We observe that HybridRAG can achieve a higher F1 score compared with the combined performance of VectorRAG and KeywordRAG, indicating superior overall performance.

\subsubsection{Retriever metrics} 
To evaluate the retrieval ability of HybridRAG, we compute Claim Recall (CR) and Context Precision (CP), where CR measures the proportion of claims involved in ground-truth answers among retrieved chunks, while CP quantifies the proportion of relevant chunks in the retrieval context. We observe that HybridRAG outperforms the combination of VectorRAG and KeywordRAG in terms of both CR and CP metrics, indicating superior retrieval performance.

\subsubsection{Generator metrics}
Generator metrics are utilized to evaluate the response quality and generation performance of HybridRAG-based LLM agents, and we adopt six generator metrics covering various aspects of response generation: Context Utilization (CU) reflecting the extent of effectively utilizing relevant information in the context, Relevant Noise Sensitivity (NS(I)) representing the proportion of incorrect claims entailed in relevant chunks, Irrelevant Noise Sensitivity (NS(II)) representing the proportion of incorrect claims entailed in irrelevant chunks, Hallucination (Hallu.) representing the proportion of incorrect claims not entailed in any retrieved chunk, Self-Knowledge (SK) representing the proportion of correct claims generated by the LLM agent, and Faithfulness (Faith.) describing the extent of using the retrieval context by the LLM agent. We observe that, in addition to Faith. and Hallu. metrics, HybridRAG possesses better performance in other generator metrics.

Overall, HybridRAG outperforms traditional RAG in the context of optimization problem formulation for multi-UAV-assisted MEC networks.

\subsection{Performance Evaluation of the Proposed Algorithm}
In Fig. \ref{Performance_compare}, we present the performance evaluation of the proposed R\textsuperscript{2}DSAC algorithm for carbon emission optimization in multi-UAV-assisted MEC networks. Specifically,  Fig. \ref{Algorithm_compare} shows the performance comparison of the R\textsuperscript{2}DSAC algorithm with the random algorithm and two conventional DRL algorithms, i.e., SAC and Proximal Policy Optimization (PPO). We observe that the R\textsuperscript{2}DSAC algorithm achieves the highest test rewards among them and converges more rapidly than both PPO and SAC algorithms. Furthermore, it exhibits a $64\%$ performance improvement over SAC. Moreover, the R\textsuperscript{2}DSAC algorithm achieves the lowest carbon emissions for task offloading in multi-UAV-assisted MEC networks. The reason is that the diffusion process of the R\textsuperscript{2}DSAC algorithm helps mitigate the effects of noise and randomness during the generation of optimal strategies. Notably, to demonstrate that the algorithm itself does not introduce significant additional carbon emissions during strategy generation for multi-UAV task offloading, we employ CodeCarbon\footnote{\url{https://github.com/mlco2/codecarbon}} and estimate that the carbon emissions generated during model training are approximately $70.3\:\rm{g}$. After completing model training, the R\textsuperscript{2}DSAC algorithm can be directly used to generate optimal strategies, with an estimated carbon emission of only $0.025\:\rm{g}$ per inference. In Fig. \ref{ablation}, we conduct ablation experiments to evaluate the effectiveness of algorithm modules. Specifically, we compare the R\textsuperscript{2}DSAC algorithm with three baseline variants: 1) Behavior-Cloning Diffusion-SAC (BCDSAC) algorithm without the dynamic pruning module; 2) Tiny Diffusion-SAC algorithm (TDSAC) without the policy improvement module; 3) Diffusion-SAC (DSAC) algorithm without both modules. We observe that the R\textsuperscript{2}DSAC algorithm outperforms the three baseline variants, achieving the highest final reward. This improvement can be attributed to two key factors: diffusion entropy regularization, which provides stable imitation learning signals to help prevent policy collapse; and action entropy regularization, which encourages the DRL agent to explore a wider range of strategies, thereby avoiding convergence to a local optimum. Overall, the above analysis demonstrates the effectiveness of the proposed algorithm.

As illustrated in Fig. \ref{Normalization_analysis}, we evaluate the impacts of pruning rate $\varrho$, diffusion step $T$, and behavior-cloning weight $\rho$ on final rewards, training time, and carbon emissions associated with model training, normalized to $[0,1]$. In each evaluation, we only adjust one parameter while fixing the others. From Fig. \ref{Pruning_normalization}, we observe that the comprehensive performance of the R\textsuperscript{2}DSAC algorithm is optimal when $\varrho = 0.1$, achieving the lowest training time and carbon emissions while maintaining strong model performance. From Fig. \ref{Diffusion_normalization}, we observe that the normalized values of both training time and carbon emissions decrease as the number of diffusion steps increases, indicating that the length of the diffusion chain significantly influences the computational overhead of model training. Moreover, the algorithm achieves its best performance when $T = 3$, yielding the highest final reward. From Fig. \ref{Behavior_normalization}, we observe that as the behavior-cloning weight increases, the DRL agent is more likely to fall into suboptimal solution spaces. This occurs because excessive reliance on behavioral cloning can lead to insufficient exploration, limiting the DRL agent from discovering optimal strategies.

\begin{figure*}[t]
		\centering
		\begin{subfigure}[Environmental state 1.]{\includegraphics[width=0.3\linewidth]{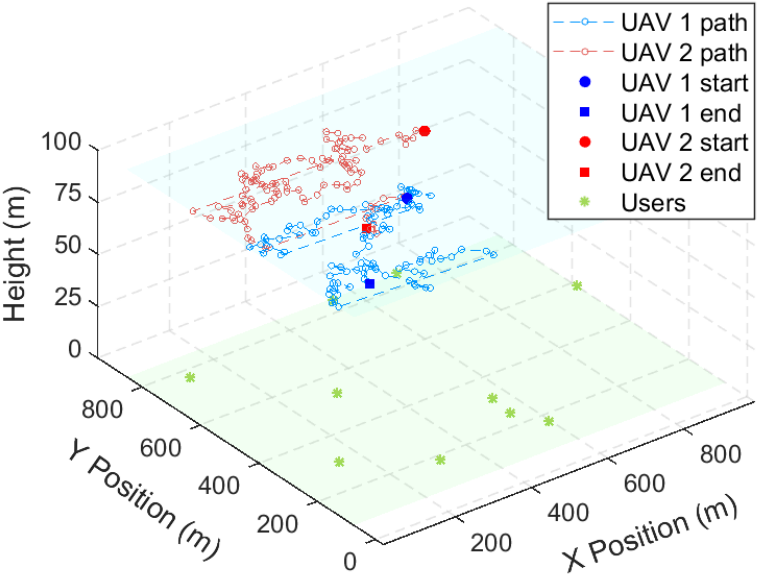}\label{uav1}}
		\end{subfigure}
		\hfill
		\begin{subfigure}[Environmental state 2.]{\includegraphics[width=0.3\linewidth]{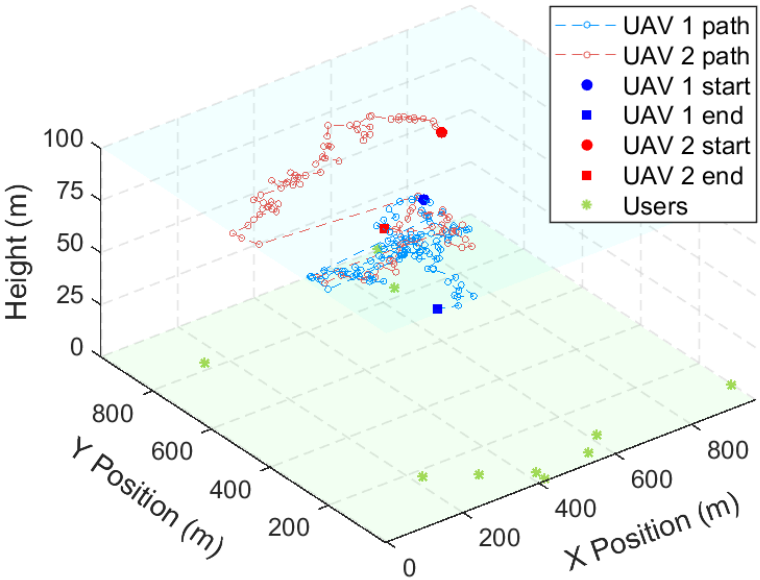}\label{uav2}}
		\end{subfigure}
		\hfill
		\begin{subfigure}[Environmental state 3.]{\includegraphics[width=0.3\linewidth]{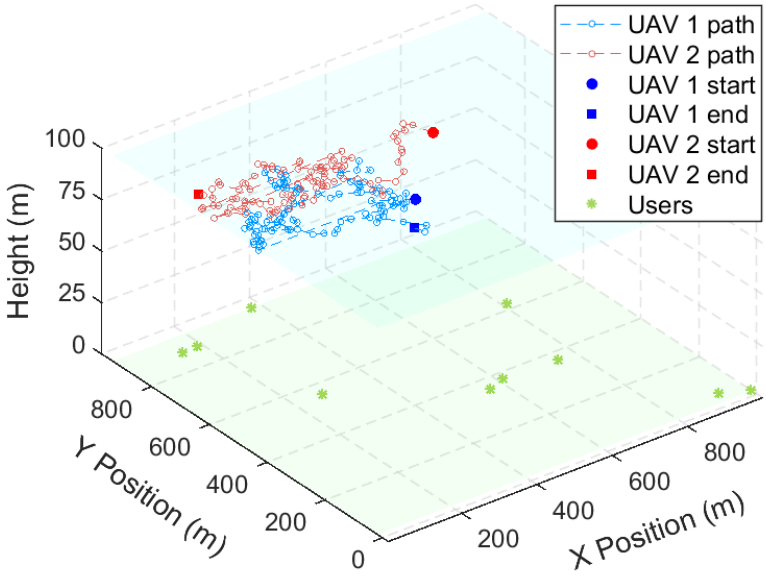}\label{uav3}}
		\end{subfigure}
		\caption{UAV trajectories generated by the R\textsuperscript{2}DSAC algorithm within a single episode under varying environmental states.}
		\label{UAV_path}
\end{figure*}

In Fig. \ref{UAV_path}, we show two UAV trajectories generated by the R\textsuperscript{2}DSAC algorithm within a single episode under varying environmental states. Each environmental state represents a specific user position set. Initially, the two UAVs depart from their respective fixed starting positions. At each time slot, the UAV manager determines the flying directions and velocities of the UAVs based on the current user request condition. We observe that, regardless of the environmental state, the end positions of the two UAVs tend to be located in areas with high user density, indicating the stability of the R\textsuperscript{2}DSAC algorithm.

\section{Conclusion}\label{Conclusion}
In this paper, we have studied the implementation of LLM agents for enabling low-carbon LAENets. Specifically, we have developed a HybridRAG-based LLM agent framework for carbon emission optimization in multi-UAV-assisted MEC networks. In this framework, we have developed HybridRAG by merging KeywordRAG, VectorRAG, and GraphRAG, enabling LLM agents to generate more precise carbon emission optimization problems through the effective retrieval of structural relational information. To solve the formulated problem, we have proposed the R\textsuperscript{2}DSAC algorithm, which incorporates diffusion entropy regularization and action entropy regularization to enhance policy learning and prevent suboptimal convergence. Furthermore, we have designed a dynamic pruning module that masks unimportant neurons in the diffusion-based actor network, thereby reducing carbon emissions during model training. Simulation results demonstrate the effectiveness of the proposed framework and algorithm. For future work, we plan to explore prompt engineering techniques to optimize the interaction process between network designers and LLM agents. Additionally, we aim to develop multi-agent diffusion model-based DRL algorithms to derive optimal strategies for carbon emission optimization in LANets.

\bibliographystyle{IEEEtran}
\bibliography{ref}
\end{document}